\documentclass[11pt]{article}
\usepackage[T1]{fontenc}
\usepackage{amsfonts}
\usepackage{amsmath}
\usepackage{enumerate}
\usepackage{amssymb}
\usepackage{amsthm}
\usepackage{bbm}
\usepackage{bm}
\usepackage{mathrsfs}
\usepackage{verbatim}
\usepackage{setspace}
\usepackage{color}
\usepackage{pdfsync}
\usepackage{enumitem}
\theoremstyle{plain}
\newtheorem{theorem}{Theorem}[section]
\newtheorem{proposition}[theorem]{Proposition}
\newtheorem{lemma}[theorem]{Lemma}

\theoremstyle{definition}

\newtheorem{remark}[theorem]{Remark}

\newtheorem{example}[theorem]{Example}

\newtheorem{assumption}[theorem]{Assumption}

\theoremstyle{remark}

\renewenvironment{thebibliography}[1]{%
\begin{oldthebibliography}{#1}%
\setlength{\baselineskip}{.9em}
\linespread{1}
\small
\setlength{\parskip}{0.3ex}%
\setlength{\itemsep}{.5em}%
}%
{%
\end{oldthebibliography}%
}

\newcommand{\E}{\mathbb{E}}
\newcommand{\F}{\mathbb{F}}

\newcommand{\N}{\mathbb{N}}
\renewcommand{\P}{\mathbb{P}}
\newcommand{\Q}{\mathbb{Q}}
\newcommand{\R}{\mathbb{R}}
\renewcommand{\S}{\mathbb{S}}

\newcommand{\cF}{\mathcal{F}}
\newcommand{\cG}{\mathcal{G}}
\newcommand{\cH}{\mathcal{H}}

\newcommand{\cM}{\mathcal{M}}

\newcommand{\cP}{\mathcal{P}}

\newcommand{\fP}{\mathfrak{P}}

\newcommand{\as}{\mbox{-a.s.}}

\newcommand{\ov}{\overline}
\numberwithin{equation}{section}

\usepackage[pdfborder={0 0 0}]{hyperref}
\hypersetup{
  urlcolor = black,
  pdfauthor = {Daniel Bartl, Michael Kupper, Ariel Neufeld},
  pdfkeywords = {T.B.D.;
},
}

\begin{document}

\title{\vspace{-5em} 
Pathwise superhedging on prediction sets
\date{\today}
\author{
  Daniel Bartl%
  \thanks{
  Department Mathematics, University of Vienna, \texttt{daniel.bartl@univie.ac.at}.
  }
  \and
    Michael Kupper%
     \thanks{
     Department Mathematics and Statistics, University of Konstanz, \texttt{kupper@uni-konstanz.de}.
     }
  \and
  Ariel Neufeld%
   \thanks{
   Division of Mathematical Sciences, NTU Singapore, \texttt{ariel.neufeld@ntu.edu.sg}.   }
 }
}
\maketitle \vspace{-1.2em}

\begin{abstract}
In this paper we provide a pricing-hedging duality for the model-independent superhedging price with respect to a prediction set $\Xi\subseteq C[0,T]$, where  the superhedging property needs to hold pathwise, but only for paths lying in $\Xi$. For any Borel measurable claim $\xi$ which is bounded from below, the superhedging price coincides with the supremum over all pricing functionals $\E_\Q[\xi]$ with respect to martingale measures $\Q$ concentrated on the prediction set $\Xi$. This allows to include beliefs in future paths of the price process expressed by the set $\Xi$, while  eliminating all those which are seen as impossible. Moreover, we provide several examples to justify our setup.
\end{abstract}

\vspace{.9em}

{\small
\noindent \emph{Keywords} Model-Independent Superhedging; Pricing-Hedging Duality; Modeling Beliefs

\noindent \emph{AMS 2010 Subject Classification}
91B24; 91G20;  60G44 
}

\section{Introduction}
In this paper we study the problem of pathwise superhedging on a prediction set $\Xi\subseteq C[0,T]$
of continuous price paths, i.e.~finding a predictable trading strategy which super-replicates a given contingent claim $\xi:C[0,T]\to\mathbb{R}$ simultaneously for all possible future price paths in $\Xi$.

Unlike the famous Black-Scholes model, most financial models cannot exactly replicate every contingent claim. This phenomenon, called incompleteness of the market is equivalent to the failure of uniqueness of equivalent local martingale measures. Since there is not an unique price which financial agents are willing to accept, the concept of superhedging starting with \cite{ElKarouiQuenez.95} has been well established in the financial literature. Here one wants to  find the smallest initial capital for which a trading strategy exists which superhedges the claim $\xi$.

To be more precise, in classical finance, one assigns probabilities to all events
by fixing a probability measure $\P$. Then, the superhedging property is required to hold $\P$-a.s.

Recently, motivated by the early works of \cite{Knight.21,Ellsberg.61}, 
one started to consider a set of probability measures $\cP$, rather than an unique one, where each element represents the candidates for the possible right law. In the so-called \textit{quasi-sure setting}, one then requires the superhedging property to hold true $\cP$-quasi surely, which means $\P$-a.s. for all $\P\in \cP$. This problem under volatility uncertainty was motivated by the early works of \cite{AvellanedaLevyParas.95,Lyons.95} and later has also been solved in \cite{DenisMartini.06,SonerTouziZhang.2010bsde,NeufeldNutz.12,PossamaiRoyerTouzi.13}.

In the \textit{model-independent} (or \textit{pathwise}) approach, one wants to go away from the classical assumption of assigning probabilities to events related to the financial market by fixing one probability measure, or a set of probability measures allowing for model ambiguity. In this setting, superhedging is required to hold true  for every possible future path in $C[0,T]$ of the price process. Such an approach has started with the seminal work \cite{Hobson.98} and has been recently lead to attention in various other works; we refer to \cite{AcciaioBeiglbockPenknerSchachermayer.16,BartlKupperProemelTangpi.17,BurzoniFrittelliMaggis.15,DolinskySoner.12}, to name but a few. 

However, it  turns out that the concept of superhedging is too robust leading to too high prices. In fact, for stochastic volatility or rough volatility models, it turns out that the classical superhedging price coincides with the model-independent one and is so high that for Markovian payoffs of the form $\Gamma(S_T)$, like e.g.\ the European Call and Put option, the optimal superhedging
strategy can be chosen to be of buy-and-hold type, see  \cite{DolinskyNeufeld.16,Neufeld.17}.  To reduce the model-independent superhedging price,  inspired by the work of \cite{Mykland.03}, \cite{HouObloj.15} introduced the concept of prediction sets, where agents may allow to exclude paths which they consider to be impossible to model future price paths. Hence they require the superhedging property only to hold true on every path in $\Xi\subseteq C[0,T]$ of their prediction set. 


Whereas the pricing-hedging duality is well-understood for the pathwise superhedging with respect to all paths in $C[0,T]$, it turns out that the problem becomes considerably more difficult when requiring the superhedging property only to hold true on the prediction set $\Xi\subseteq C[0,T]$.
To illustrate the difficulty, consider the examples where the agent may believe in the Black-Scholes model, or is uncertain about the volatility like in the G-expectation (see \cite{Peng.07}) and hence models his/her beliefs by requiring
\begin{equation*}
\Xi_{BS}:=\big\{ \omega \in C_0[0,T]:\  d\langle \omega \rangle_t =\sigma^2 \omega(t)^2\,dt\big\}, \quad\mbox{or} \quad  \Xi_{G}:=\big\{ \omega \in C_0[0,T]:\, \tfrac{d\langle \omega \rangle_t}{dt} \in [\underline \sigma^2,\overline \sigma^2]\big\}.
\end{equation*}
Observe that these sets are neither closed, nor $\sigma$-compact.
In \cite{HouObloj.15}, they get an asymptotic pricing-hedging duality result. More precisely, the asymptotic price being defined as the limit (when  $\varepsilon \to 0$) of superhedging prices on $\varepsilon$-varied prediction sets $\Xi^\varepsilon$ turns out to  coincide with the limit of the supremum of the pricing functionals with respect to the martingale measures having support on the $\varepsilon$-varied prediction sets $\Xi^\varepsilon$. However, typically, $\Xi$ is not closed and hence $\Xi^\varepsilon$ might be far away from the original set $\Xi$. Indeed, one can show that in the canonical example of the paths $\Xi_{BS}$ of the Black-Scholes model, $\Xi^\varepsilon=C[0,T]$ for any $\varepsilon>0$.

In \cite{BartlKupperProemelTangpi.17}, they obtain a superhedging duality with respect to a prediction set $\Xi \subseteq C[0,T]$, where the superhedging price coincides with the supremum over all pricing functionals with respect to martingale measures concentrated on the prediction set $\Xi$.  
 As trading strategies, they use simple strategies and define the gain process to be the limit inferior of the discrete integral with respect to the simple strategies; we refer to \cite{PerkowskiPromel.16,BeiglbockCoxHuesmannPerkowskiProemel.17,PerkowskiPromel.16,Vovk.12,Vovk.16} which also applied this setup in the context of superhedging. 
 However, they need to impose the crucial assumption that the prediction set $\Xi$ is $\sigma$-compact in a topology which is at least as fine as the usual sup-norm; a property which is in general not satisfied in the examples of paths motivated by financial applications, like e.g. $\Xi_{BS}$ or $\Xi_{G}$.


In this paper, we extend the work of \cite{HouObloj.15,BartlKupperProemelTangpi.17}. 
We do not require any {strong} topological properties on $\Xi$ such that our  pricing-hedging duality also covers e.g., $\Xi_{BS}$ or $\Xi_{G}$ as examples.  
{For our first results, which are stated in Theorem~\ref{thm:main.stopped}, Theorem~\ref{thm:main.stopped.generlized.integrals}, and Theorem~\ref{thm:main.with.Z},
we stick to the formulation of \cite{BartlKupperProemelTangpi.17} for the superhedging price. In Theorem~\ref{thm:main.stopped} and Theorem~\ref{thm:main.with.Z}, we allow to trade next to the stock $S$ also in the iterated integral $\mathbb{S}:=\int S \,dS$, which seems to be a bit artificial in the context of superhedging at first glance. However, we point out that the $d\mathbb{S}$-integral should not be seen as an artificial traded asset but rather as allowing
a larger set of admissible trading strategies. Indeed when enlarging the set of admissible strategies adequately, one directly obtains a superhedging duality result in the formulation of \cite{BartlKupperProemelTangpi.17} where one is only allowed to trade in the stock $S$; see Theorem~\ref{thm:main.stopped.generlized.integrals}.}

{Note that the superhedging duality results in Theorem~\ref{thm:main.stopped}, Theorem~\ref{thm:main.stopped.generlized.integrals}, and Theorem~\ref{thm:main.with.Z} hold for claims $\xi$ which are semicontinuous. We then continue the idea of enlarging the set of admissible strategies by following an idea of \cite{Vovk.16}. 
The enlargement of the admissible strategies has the consequence that our pricing functional becomes enough regular such that, with an application of Choquet's capacitability theorem in the functional form, we derive our desired superhedging duality result also for Borel measurable claims; we refer to Theorem~\ref{thm:main.stopped.liminf.clsoure} for our main result. 
}

{Summing up, our contribution is twofold. First, we derive a pathwise superhedging duality where the superhedging property only needs to hold for paths lying in a given prediction set $\Xi$ which does not require any strong topological properties. In particular, compared to previous results in the literature, we are now able to establish 
a duality result for prediction sets which enforces conditions to hold both on the price path and its quadratic variation; see Remark~\ref{rem:stopped-admissible} for a more detailed explanation also in relation to Assumption~\ref{ass:A}. Second, we obtain a pathwise superhedging duality for measurable claims. This extension from semicontinuous claims to Borel measurable ones is to the best of our knowledge the first one in the context of pathwise superhedging duality and allows to consider financial derivatives such as, e.g., digital options or financial derivatives where the regularity is not known, for example American options evaluated at optimal exercise times (see also \cite{NutzZhang.15}).}

%

{To remove strong topological requirements on the prediction set $\Xi$}, we lift the superhedging problem to the product space $\ov \Omega:=\Omega\times \Omega$, where  the first coordinate represents the original price process and the second one represents its quadratic variation; 
we refer to \cite{Haussmann.86,TanTouzi.11,LiuNeufeld.16} where similar enlarged spaces $\ov \Omega$ were considered. 
On the enlarged space $\ov  \Omega$,  we then prove  the desired superhedging duality and can then conclude the desired result on the original space. This works, roughly speaking, by observing that for any probability measure $\ov\Q$ on the enlarged space for which the first coordinate is a local martingale with the second one as its quadratic variation, the $\ov \Q$-completed natural filtration coincides with the $\ov \Q$-completed one generated only by the first coordinate. This then leads to a one-to-one correspondence to the original space.  
{For a further discussion explaining the advantage of lifting the original problem to an enlarged one, we refer to Remark~\ref{rem:stopped-admissible}.}
We point out that our trading strategies are defined with respect to the (right-continuous) natural filtration, without any completion with respect to a probability measure, such that we retain the framework of pathwise superhedging without any probabilistic beliefs. 

The remainder of this paper is organized as follows. In Section~\ref{sec:setup.and.main},
we introduce the setup and state our main results of this paper. Then in Section~\ref{sec:examples}, we provide several examples to motivate our theorems. In Section~\ref{sec:Proof}, we provide the proof of our main results. Finally, in Section~\ref{sec:technical-Results}, we attach some technical results required in the proof of our theorems.

\section{Setup and main results}
\label{sec:setup.and.main}

\subsection{Setup}
\label{subsec:setup}

Fix a finite time horizon $T\in(0,\infty)$, and let $C[0,T]$ be the space of all continuous paths $\omega\colon[0,T]\to\mathbb{R}$, which as usual is endowed with the sup-norm $\|\omega\|_\infty:=\sup_{0\leq t \leq T} |\omega(t)|$. Denote by $S=(S_t)_{0\leq t \leq T}$ the canonical process $S_t(\omega)=\omega(t)$, and define for each $m\geq1$ the sequence $\sigma^m_0:=0$,
  \[
    \sigma^m_{k+1}:=\inf\big\{t\geq \sigma^m_k : |S_t-S_{\sigma^m_k}|\geq 2^{-m}\big\},\quad k\geq0.
  \]
  Since $S$ has continuous paths, $\lim_{k\to \infty} \sigma^m_k(\omega)=\infty$ holds for all $\omega\in C[0,T]$.
  Moreover, let $\mathbb{S}\colon[0,T]\times C[0,T]\to\mathbb{R}$ be the process defined by  
  \[
  \mathbb{S}_t:=\liminf_{m\to\infty}\mathbb{S}^m_t
  \quad\text{where}\quad
    \mathbb{S}^m_t:= \sum_{k=0}^\infty S_{\sigma^m_k}\big(S_{\sigma^m_{k+1}\wedge t}-S_{\sigma^m_k\wedge t}\big).
  \]
Define a pathwise quadratic variation
$\langle\cdot\rangle\colon C[0,T] \to C[0,T]$ by 
\[\langle\omega\rangle:=\begin{cases} S^2(\omega)-S^2_0(\omega)-2\mathbb{S}(\omega) &\mbox{if }\omega\in \Omega \\ 0 &\mbox{else,}\end{cases}\]
where $\Omega$ is the Borel set of all $\omega\in C[0,T]$ such that $\mathbb{S}^m(\omega)\to\mathbb{S}(\omega)$ in the sup-norm
and $ S^2(\omega)-S^2_0(\omega)-2\mathbb{S}(\omega)$ is nondecreasing.
The space $\Omega$ is endowed with the relative topology
and equipped with the corresponding relative Borel $\sigma$-field $\cF$. Moreover, we denote by $\fP(\Omega)$ the set of Borel probability measures on $(\Omega,\cF)$. Furthermore, for any $\omega \in \Omega$ and $t \in[0,T]$ we denote by $\omega^t(s):= \omega(t\wedge s)$, $s \in [0,T]$, the stopped path of $\omega$ at time $t$.
{Finally, denote by $\F=(\cF_t)_{0\leq t \leq T}$ the raw filtration generated by the canonical process $S$ on $\Omega$, i.e.\ $\cF_t= \sigma(S_s, s\leq t)$, and by $\F_+$ its right-continuous version, $\cF_{t+}=\cap_{s>t}\cF_{s\wedge T}$, for each $t$. }

\begin{remark}
  \label{rem:Omega}
  The construction of the pathwise quadratic variation $\langle\cdot\rangle$ is similar to
  \cite{NeufeldNutz.13a} and goes back to \cite{Karandikar.95,NutzSoner.10}. For every $\Q\in\mathfrak{P}(\Omega)$ under which the canonical process is a semimartingale in the raw filtration, it is a consequence of the Burkholder-Davis-Gundy inequalities that
    \[
      \sup_{0\leq t\leq T}\bigg| \mathbb{S}^m_t-\sideset{^{(\Q)\hspace{-7pt}}}{}{\int_0^t}S_{s}\,dS_s\bigg|\to 0\quad \Q\as.
    \]
  Hence, $\langle S\rangle =S^2-S_0^2 - 2\,{}^{(\Q)\hspace{-5pt}}\int S\,dS  = \langle S\rangle^{(\Q)}$ holds $\Q\as$ as an application of the integration-by-parts formula for the It\^o integral. In particular, $\Q(\Omega)=1$.
{Here $^{(\Q)}\!\!\int$ and $\langle \cdot\rangle^{(\Q)}$ are the usual stochastic integral and quadratic variation, respectively, defined under the semimartingale meausre $\mathbb{Q}$.} 
  Notice that $\mathbb{S}_t$ coincides with F\"ollmer's pathwise stochastic integral \cite{Follmer.81} of $\int_0^t S\,dS$ on $\Omega$.
\end{remark}

Let $\mathcal{H}$ be the set of all simple processes
$H\colon [0,T]\times \Omega\to\mathbb{R}$ of the form 
$H=\sum_{l=1}^{L}  h_l1_{( \tau_l,\tau_{l+1}]}$
where $L\in\mathbb{N}$, $0\leq \tau_1\leq\dots\leq \tau_{L+1}\leq T$ are stopping times 
w.r.t.\ the filtration $\F_+$, and $ h_l\colon \Omega\to\mathbb{R}$ 
are bounded $\cF_{\tau_l+}$-measurable functions. 
For any $H\in\mathcal{H}$ the pathwise stochastic integral
\[ 
  (H\cdot S)_t({\omega}):=\sum_{l=1}^{L}  h_l({\omega}) ( S_{ \tau_{l+1}({\omega})\wedge t}({\omega})- S_{ \tau_{l}({\omega})\wedge t}({\omega}))
\]
is well-defined for all $t\in[0,T]$ and all ${\omega} \in \Omega$. 
Similarly the integral $(H\cdot \mathbb{S})_t$ is well-defined.
%
\subsection{Superhedging duality for prediction sets closed under stopping}
\label{subsec:result-stopped}
Our goal is to identify the pathwise superhedging price when the superhedging property only needs to hold for a given prediction set $\Xi\subseteq \Omega$ of price paths. We will mostly work with a prediction set which satisfies the following assumptions.
\begin{assumption}\label{ass:A}
$\Xi \subseteq \Omega$ is a nonempty set of paths of the form
\begin{equation*}
\Xi = \big\{\omega \in \Omega\colon (\omega,\langle w \rangle)\in \overline{\Xi}\big\}
\end{equation*}
for a nonempty set $\overline{\Xi} \subseteq C[0,T]\times C[0,T]$ satisfying

\vspace*{0.15cm}
\noindent
\textbf{(A1)} $\overline{\Xi}$ is the countable union of compact sets,\\
\textbf{(A2)} $(\omega,\nu) \in \overline{\Xi}$ implies that $\nu(0)=0$ and $\nu$ is nondecreasing,\\
\textbf{(A3)} for any $t \in [0,T]$ we have that $(\omega,\nu)\in \overline{\Xi}$ implies that $(\omega^t,\nu^t) \in \overline{\Xi}$.
\end{assumption} 
{We discuss these assumptions below, see Remark~\ref{rem:stopped-admissible}.}
%
To identify the pathwise superhedging price with respect to a given prediction set $\Xi\subseteq \Omega$, we denote by $\mathcal{M}(\Xi)$ the set of local martingale measures for $S$ concentrated on $\Xi$, i.e.\
\begin{equation*}
\mathcal{M}(\Xi):= \big\{\Q \in \fP(\Omega)\colon \mbox{$S$ is a $\Q$-$\F$-local martingale and } \Q(\Xi)=1 \big \}.
\end{equation*}
Notice that $\Xi \subseteq \Omega$ satisfying Assumption~\ref{ass:A} ensures that $\mathcal{M}(\Xi)$ is nonempty, as condition (A3) enforces $\Xi$ to contain constant paths.
Our first main result is the following.
\begin{theorem}
\label{thm:main.stopped}
Let $\Xi\subseteq \Omega$ be a prediction set satisfying Assumption~\ref{ass:A}. Then
  \begin{align*}
   \Phi(\xi):&= \inf\left\{ \lambda\in\mathbb{R} : 
    \begin{array}{ll}
    \text{there are sequences $(H^n)$ and  $(G^n)$ in ${\mathcal{H}}$ such that }\\
     \lambda + (H^n\cdot S)_t+(G^n\cdot  \mathbb{S})_t\geq 0 \text{ on } \Xi \text{ for all } n,t \\
  \lambda+ \liminf_{n \to \infty}\big( (H^n\cdot S)_T +( G^n\cdot  \mathbb{S})_T \big)\geq  \xi \text{ on } \Xi
  \end{array}
  \!\right\}\\
  &= \sup_{\mathbb{Q}\in\mathcal{M}(\Xi)} \mathbb{E}_\mathbb{Q}  [\xi]
  \end{align*}
  for every function $\xi\colon C[0,T]\to[0,\infty]$ of the form $\xi(\omega)=\liminf_{n\to \infty}\ov\xi_n(\omega,\langle\omega\rangle)$
  where $\ov\xi_n\colon C[0,T]\times C[0,T]\to[0,\infty)$ are bounded and upper-semicontinuous {(in particular, every upper/lower-semicontinuous function $\xi$).}
\end{theorem}
%
%

\begin{remark}
\label{rem:liminf}
  Using gain processes of the form $\liminf_n H^n\cdot S$ for a sequence of simple integrands $(H^n)$ can be seen as the pathwise analogue of the classical gain process in mathematical finance being a stochastic integral with respect to $dS$ (under a given measure $\P$). 
  Indeed, the construction of the classical stochastic integral is accomplished by a $L^2(\P)$-limit procedure with respect to simple integrands. 
  We refer to \cite{PerkowskiPromel.16} for further discussions regarding pathwise stochastic integrals applied to model-independent finance.

{
  Nevertheless, an evident question in the present context is whether the $\liminf$ in the definition of $\Phi$ is actually needed.
  While the results of \cite{DolinskySoner.12,HouObloj.15} suggest that for uniformly continuous $\xi$ this might not be the case (at least without or with very ``regular'' prediction sets), we provide at the end of Section~\ref{sec:examples} an  example
  showing the necessity of the $\liminf$ in the present setting; we refer to Example~\ref{ex:duality-gap} for further details.}
\end{remark}

In Theorem~\ref{thm:main.stopped}, we do not only allow such gain processes with respect to $dS$, but also  with respect to $d\mathbb{S}$. 
At first glance, this might look artificial from a financial point of view. 
However, the $d\mathbb{S}$-integral in the definition of $\Phi$ should not be seen as an artificial traded asset but rather as allowing a larger set of admissible trading strategies. 
Indeed, by F\"ollmer~\cite{Follmer.81}, $\mathbb{S}=(S\cdot S)$ can be defined pathwise on $\Omega$ which implies that $(G\cdot  \mathbb{S})=(GS\cdot S)$ can also be defined pathwise on $\Omega$ for $G\in\mathcal{H}$.
Thus, instead of allowing to trade in $S$ and $\mathbb{S}$ with simple strategies, one could also allow to trade only in $S$, but with those integrands for which the integral can be defined pathwise and coincides with the classical integral under every martingale measure on $\Omega$ (or, at least, integrands of the form $H+GS$ for $H$ and $G$ simple). 
{In particular, denoting by $\mathcal{H}^{{\mathrm{gen}}}:=\{ H + GS : H,G\in\mathcal{H}\}$ the set of generalized simple processes, this implies that Theorem \ref{thm:main.stopped} can be restated as follows:}
{
\begin{theorem}
\label{thm:main.stopped.generlized.integrals}
Let $\Xi\subseteq \Omega$ be a prediction set satisfying Assumption~\ref{ass:A}. Then
  \begin{align*}
   \Phi^{{\mathrm{gen}}}(\xi):=\inf\left\{ \lambda\in\mathbb{R} : 
    \begin{array}{ll}
    \text{there is a sequence $(H^n)$ in ${\mathcal{H}^{\mathrm{gen}}}$ such that }\\
     \lambda + (H^n\cdot S)_t\geq 0 \text{ on } \Xi \text{ for all } n,t \\
  \lambda+ \liminf_{n \to \infty}(H^n\cdot S)_T \geq  \xi \text{ on } \Xi
  \end{array}
  \!\right\}
  &= \sup_{\mathbb{Q}\in\mathcal{M}(\Xi)} \mathbb{E}_\mathbb{Q}  [\xi]
  \end{align*}
  for every function $\xi\colon C[0,T]\to[0,\infty]$ of the form $\xi(\omega)=\liminf_{n\to \infty}\ov\xi_n(\omega,\langle\omega\rangle)$
  where $\ov\xi_n\colon C[0,T]\times C[0,T]\to[0,\infty)$ are bounded and upper-semicontinuous.
\end{theorem}
}

  {One can take this idea even one step further: 
  Notice that $(G\cdot \mathbb{S})=(GS\cdot S)$ is the (uniform) limit of the sequence $(GS^m\cdot S)$, where $S^m:=\sum_{k=0}^\infty S_{\sigma^m_k} 1_{(\sigma^m_{k},\sigma^m_{k+1}]}$ with $(\sigma^m_k)$ being the stopping  times from Subsection \ref{subsec:setup}. 
  In other words, $(G\cdot  \mathbb{S})$ is the limit of a sequence of $dS$-integrals with respect to this intermediate class of integrands.
  Therefore, taking this class of integrands as admissible trading strategies, we found yet another reformulation of Theorem \ref{thm:main.stopped} (which we do not state explicitly).
  In Theorem~\ref{thm:main.stopped.liminf.clsoure} below we work out the idea of enlarging the trading strategies once more, but with the goal to obtain duality for Borel measurable claims.}
\begin{remark}
\label{rem:stopped-admissible}
{
  %
  There are not many different strategies to prove model-free superhedging duality in mathematical finance. One of the first techniques developed in  \cite{DolinskySoner.12}, where no prediction sets were involved, is to approximate the duality problem by a  discretization in time and space to reduce  the original duality problem to ones in finite-dimensions, where one can  apply classical duality theorems to derive the result.
  %
  However, it seems that when restricting the superhedging property to hold true only on a given prediction set, that this technique comes at the price of enforcing a lot of regularity on the prediction set  to obtain duality; see \cite{HouObloj.15}.
  Another  established strategy developed in \cite{AcciaioBeiglbockPenknerSchachermayer.16} and its many successors is what we refer to the topological approach, where 
  to overcome technicalities implied by the absence of a reference probability measure  requires continuous functions to be rich enough to determine the behavior of the functional, as well as some compactness of the problem.
  More concretely, this method requires the prediction set $\Xi$ to be $\sigma$-compact, i.e.\
  the countable union of compacts.
  Under the requirement of $\sigma$-compactness of the prediction set, \cite{BartlKupperProemelTangpi.17} provides a pathwise superhedging duality  in the present continuous-time setting.
  {However, for prediction sets involving the quadratic variation (see Section~\ref{sec:examples}), $\sigma$-compactness generally fails or cannot be verified, since
   the map
   $C[0,T] \ni\omega \mapsto \langle \omega \rangle \in C[0,T]$  is highly not continuous.  
  The immediate idea of changing the topology and making $\omega\mapsto \langle\omega\rangle$ continuous does not work either as compact sets in this topology are very ``small'' and one looses separability of $C[0,T]$.
}
}

  {
  To overcome this problem, we shall work in the enlarged space $\ov\Omega :=C[0,T]\times C[0,T]$ where the first component will play the role of the price process $S$ and the second component will play the role of its quadratic variation. 
  On this enlarged space, assumption (A1) (which requires the lifted prediction set $\ov{\Xi}$ to be the countable union of compacts) implies that the afore mentioned topological proof works on the enlarged space. 
  The dual objects $\ov{\mathbb{Q}}$ are now martingale measures for the first and second component, and assumption (A2) ensures that every martingale measure for the first component on the enlarged space will have the second component as its quadratic variation. 
  This is crucial to go back and forth between the original space and the enlarged space.}

 {Only condition (A3) seems artificial at first and in fact has nothing to do with the above discussion.
However, it clears the way for the (nicer) formulation of the superhedging price $\Phi$, compared to the superhedging price $\Phi^Z$ in Theorem \ref{thm:main.with.Z} where (A3) is not assumed. Moreover, it naturally provides a relation to American options, see Remark \ref{rem:ass-stopped-american}.
In addition, due to condition (A3), the admissibility condition imposed in the definition of $\Phi$ could be replaced by $\lambda+(H^n\cdot S)_T+(G^n\cdot  \mathbb{S})_T\geq 0$, see Lemma~\ref{lem:integral.positive} for the precise statement.}
  
\end{remark}
\begin{remark}
\label{rem:ass:nonneg-main-stopped}
  The nonnegativity assumption in Theorem~\ref{thm:main.stopped} imposed on $\xi$ could be relaxed by instead requiring $\xi\colon C[0,T] \to [-\infty,\infty]$ to be bounded from below on $\Xi$, if the admissibility condition $\lambda +(H^n\cdot S)_t+(G^n\cdot  \mathbb{S})_t\geq 0$ in the definition of $\Phi$ is replaced by $\lambda +(H^n\cdot S)_t+(G^n\cdot  \mathbb{S})_t\geq \inf_{\omega\in\Xi}\xi(\omega)$.
\end{remark}
\begin{remark}
\label{rem:ass-stopped-american}
  Condition (A3) in Assumption~\ref{ass:A} implies for every Borel measurable function $\xi\colon C[0,T] \to (-\infty,\infty]$ which is bounded from below the following identity
  \begin{equation*}
  \sup_{\Q\in\mathcal{M}(\Xi)} \mathbb{E}_\mathbb{Q}[\xi]
  =\sup_{\tau \text{ is }\mathbb{F}\text{-stopping time }} \sup_{\Q\in\mathcal{M}(\Xi)} \mathbb{E}_\mathbb{Q}[\xi(S^\tau)].
  \end{equation*}
  Indeed,  by choosing $\widehat{\tau}=T$, it follows that the left-hand side is smaller than or equal to the right hand side. 
  As for the reverse inequality, observe that for any $\mathbb{F}$-stopping time $\tau$ and $\Q\in\mathcal{M}(\Xi)$ one has that $\Q^\tau:=\Q\circ (S^\tau)^{-1}$ defines a local martingale measure for $S$ which by condition (A3)  also satisfies $\Q^\tau(\Xi)=1$. 
  Hence $\Q^\tau \in \mathcal{M}(\Xi)$, which in turns implies the reverse inequality.

  Moreover, for every given $\Xi$ satisfying conditions (A1)-(A2), the set
  \[\Xi^{\mathrm{stop}}:=\big\{\omega^t : \omega\in\Xi \text{ and }t\in[0,T]\big\}\]
  satisfies (A1)-(A3).
  Indeed, since for all $t\in[0,T]$ we have by construction of the pathwise quadratic variation that $\langle\omega\rangle^t=\langle\omega^t\rangle$, we see that 
  \[\ov\Xi^{\mathrm{stop}}:=\big\{(\omega^t,\nu^t) : (\omega,\nu)\in\ov \Xi \text{ and }t\in[0,T]\big\}\]
  satisfies $\Xi^{\mathrm{stop}}=\{\omega\in\Omega: (\omega,\langle\omega\rangle)\in\ov\Xi^{\mathrm{stop}}\}$.
  By construction, $\ov\Xi^{\mathrm{stop}}$ satisfies (A2)-(A3).
  To see (A1), let $\ov\Xi_n$ be compact sets such that $\ov\Xi=\bigcup_n\ov\Xi_n$. 
  Then, as the map 
  $\mathrm{stop}\colon C[0,T]\times[0,T]\to C[0,T]$, $(\omega,t)\mapsto \omega^t$ is continuous, one has $\ov\Xi^{\mathrm{stop}}=\bigcup_n\mathrm{stop}(\ov\Xi_n\times[0,T])$ and  $\mathrm{stop}(\ov\Xi_n\times[0,T])$ are compact sets, as the continuous image of compact sets.
\end{remark}
We continue the previous idea of enlargement of the set of trading strategies in the definition of the superhedging price. 
More precisely, following an idea of Vovk \cite{Vovk.16}, we define for every $\Xi \subseteq \Omega$ and $\lambda\geq0$ the set
\begin{equation*}\mathcal{G}^{\Xi}_\lambda:=\liminf\text{-closure of }\big\{\lambda + (H\cdot S)_T: H\in\mathcal{H}
\text{ and } \lambda+(H\cdot S)_t\geq 0 \text{ on } \Xi \text{ for all }t \big\},
\end{equation*}
i.e.~$\mathcal{G}^\Xi_\lambda$ is the smallest set of functions $X\colon C[0,T] \to [-\infty,\infty]$ which contains $\lambda + (H\cdot S)_T$ for $H\in\mathcal{H}$
with $\lambda+(H\cdot S)_t\geq 0$ on $\Xi$ for all $t$,
such that $\liminf_n X_n\in\mathcal{G}^\Xi_\lambda$ whenever $X_n\in\mathcal{G}^\Xi_\lambda$ for every $n$. 
We refer to \cite{Vovk.16} for more details regarding the liminf closure.  
Then we obtain the following result.
\begin{theorem}
\label{thm:main.stopped.liminf.clsoure}
  Let $\Xi\subseteq \Omega$ be a  prediction set satisfying Assumption~\ref{ass:A}.
  Then it holds that
  \begin{align*}
  \Phi^{\mathrm{cl}}(\xi):&=\inf\big\{ \lambda\in\mathbb{R} : 
    \text{there is } Y\in\mathcal{G}^\Xi_\lambda \text{ such that } Y\geq \xi \text{ on } \Xi
  \big\} \\
  &=  \sup_{\mathbb{Q}\in\mathcal{M}(\Xi)} \mathbb{E}_\mathbb{Q}  [\xi]
  \end{align*}
  for every Borel function $\xi\colon C[0,T]\to[0,\infty]$.
\end{theorem}
{As in Remark \ref{rem:ass:nonneg-main-stopped}, the non-negativity of $\xi$ can be relaxed.}
%
\begin{remark}\label{rem:duality-meas-large}
{
To the best our knowledge, this is the first pathwise superhedging duality result in a continuous-time setting for measurable claims. Hence  the result is not only of interest 
for a ``small'' prediction set, but  also for 
``large'' ones, e.g.,~the set $\Xi$ of all paths which possess a quadratic variation and are H\"older continuous. During the reviewing process of this work, a similar duality theorem in the context of martingale optimal transport on the space of c\`adl\`ag paths has been obtained in \cite{cheridito2019martingale}.
Under a suitable constructed topology on the space of c\`adl\`ag paths, the duality result in \cite{cheridito2019martingale}  has also been established by  enlarging the set of admissible strategies using the $\liminf$- closure in the spirit of \cite{Vovk.16} and then extending the duality result from semicontinuous claims to Borel measurable ones following the approach in \cite{Kellerer84} to apply  Choquet's capacitability theorem.}
\end{remark}


  

\begin{remark}
\label{rem:quasi-sure}
  As a direct consequence it follows that a Borel set $A \subseteq C[0,T]$ satisfies $\mathbb{Q}(A)=0$ for all $\mathbb{Q}\in\mathcal{M}(\Xi)$ if and only if for every $\varepsilon>0$ there exists $Y^\varepsilon\in\mathcal{G}^\Xi_\varepsilon$ such that $1_A\leq Y^\varepsilon$.
\end{remark}

\begin{remark}
\label{rem:weak.duality}
  For any prediction set $\Xi$, trading strategy $H\in\mathcal{H}$, and martingale measure $\mathbb{Q}$, it holds that  $\mathbb{E}_{\mathbb{Q}}[\lambda + (H\cdot S)_T] =\lambda$ for all $\lambda\in\mathbb{R}$.
  Fatou's lemma implies that for every sequence $H^n\in\mathcal{H}$ such that
  $\lambda + (H^n\cdot S)_t\geq 0$ on $\Xi$ for every $n$ and $t$, one has 
  $\mathbb{E}_{\mathbb{Q}}[\lambda + \liminf_n(H^n\cdot S)_T]\leq \lambda$.
  This immediately implies weak duality, i.e.~that $\Phi(\xi)\geq \sup_{\Q\in\mathcal{M}(\Xi)} \mathbb{E}_{\Q}[\xi]$ for every Borel $\xi\colon C[0,T]\to[0,\infty]$.

  As for the liminf-closure, fix some $\lambda\in\mathbb{R}$ and notice that the set 
  $\mathcal{A}:=\{X \geq 0 : \mathbb{E}_\mathbb{Q}[X]\leq \lambda \text{ for all }\Q\in\mathcal{M}(\Xi)\}$ contains $\lambda + (H\cdot S)_T$ for $H\in\mathcal{H}$ such that
  $\lambda + (H\cdot S)_t\geq 0$ on $\Xi$ for every $t$.
  As $\mathcal{A}$ is closed under liminf (by Fatou's lemma), one has $\mathcal{G}^{\Xi}_\lambda\subseteq\mathcal{A}$. This again implies weak duality, 
  i.e.~$\Phi^\mathrm{cl}(\xi)\geq\sup_{\Q\in\mathcal{M}(\Xi)} \mathbb{E}_\mathbb{Q}[\xi]$
  for every Borel $\xi\colon C[0,T]\to[0,\infty]$.
\end{remark}

\subsection{Superhedging duality for prediction sets not closed under stopping}
\label{subsec:result-general}
In this subsection we seek to characterize the pathwise superhedging price with respect to a given prediction set $\Xi \subseteq \Omega$ which is not necessarily closed under stopping. To that end we impose the following conditions on the prediction set.
 
 \begin{assumption}\label{ass:B}
 	$\Xi \subseteq \Omega$ is a nonempty set of paths of the form
 	\begin{equation*}
 	\Xi = \big\{\omega \in \Omega\colon \overline{Z}(\omega,\langle \omega \rangle)<\infty\big\}
 	\end{equation*}
 for some function $\overline{Z}\colon C[0,T]\times C[0,T]\to [1,\infty]$ which satisfies that
 	
 	\vspace*{0.15cm}
 	\noindent
 	\textbf{(B1)} $\overline{Z}$ has compact sublevel sets,\\
 	\textbf{(B2)} for all $(\omega,\nu)$ we have that  $\overline{Z}(\omega,\nu)<\infty$ implies that $\nu(0)=0$ and $\nu$ is nondecreasing,\\
 	\textbf{(B3)} $\ov Z(\omega,\nu)\geq \|\omega\|_\infty + \|\nu\|_\infty$ for all $(\omega,\nu)$.
 \end{assumption} 

We denote for any given set of paths $\Xi \subseteq \Omega$ and any Borel function $Z\colon \Omega \to [1,\infty]$ the set
$\mathcal{M}_{Z}(\Xi)$ of Borel probability measures defined by
\begin{equation*}
\mathcal{M}_Z(\Xi):= \big\{\Q \in \fP(\Omega)\colon \mbox{$S$ is a $\Q$-$\F$-local martingale},\, \Q(\Xi)=1,\text{ and } \E_{\Q}[Z]<\infty \big \}.
\end{equation*}

\begin{theorem}
\label{thm:main.with.Z}
Let $\Xi\subseteq \Omega$ be a prediction set satisfying Assumption~\ref{ass:B}, let 
$Z\colon \Omega \to [1,\infty]$ be the function defined by
$Z(\omega):=\overline{Z}(\omega,\langle\omega\rangle)$, 
$\omega \in \Omega$, and assume that $\mathcal{M}_Z(\Xi)$ is nonempty.
  Then   
  \begin{align*}
    \Phi^Z(\xi):&=\inf\left\{ \lambda\in\mathbb{R} : 
    \begin{array}{l}
    \text{there is  $c\ge 0$ and sequences $( H^n)$, $(G^n)$ in ${\mathcal{H}}$}
    \text{ such that}\\
    \lambda+( H^n\cdot S)_T+(G^n\cdot  \mathbb{S})_T\geq -cZ \text{ on } \Xi\text{ for all n},\\
   \lambda+ \liminf_{n \to \infty}\big( (H^n\cdot S)_T +( G^n\cdot  \mathbb{S})_T \big)\geq  \xi \text{ on } \Xi
    \end{array}
    \!\right\}\\
    &= \sup_{\Q\in\mathcal{M}_{Z}(\Xi)} \E_\Q[\xi]
    \end{align*}
    for every function $\xi\colon C[0,T] \to(-\infty,\infty]$ bounded from below which is  of the form $\xi(\omega)=\liminf_{n\to \infty}\ov\xi_n(\omega,\langle\omega\rangle)$
    where $\ov\xi_n\colon C[0,T]\times C[0,T]\to\R$ are bounded and upper-semicontinuous.
    Moreover, if $\Phi^Z(\xi)<\infty$, then the infimum is attained.
\end{theorem}
\begin{remark}
\label{rem:comaprison-Ass-A-B}
The necessity in introducing the growth function $Z$ (induced by $\overline{Z}$) is a purely technical feature and lies in the admissibility condition required to ensure that gain processes are supermartingales.

Observe that Assumption~\ref{ass:B} implies the conditions (A1)-(A2) of Assumption~\ref{ass:A}. Indeed, Assumption~\ref{ass:B} ensures that $\Xi=\{\omega \in \Omega\colon (\omega,\langle \omega \rangle) \in \overline \Xi\}$ for
\begin{equation*}
\overline{\Xi}:= \big\{(\omega,\nu) \in C[0,T]\times C[0,T]\colon \overline{Z}(\omega,\nu)<\infty\big\}.
\end{equation*}

Conversely, given a prediction set $\Xi\subseteq \Omega$ of the form $\Xi=\{\omega \in \Omega\colon (\omega,\langle \omega \rangle) \in \overline \Xi\}$ for some $\ov\Xi \subseteq C[0,T]\times C[0,T]$,
 one may ask for sufficient conditions such that  Assumption~\ref{ass:B} is satisfied. To that end, notice first that if 
$\ov\Xi$ is compact one can  set $\overline{Z}(\omega,\nu):=c+\infty 1_{\Xi^c}(\omega)$ for a suitable constant $c$ to ensure that Assumption~\ref{ass:B} is satisfied.
Having this in mind, the strategy to find the corresponding $\overline{Z}$
when $\ov\Xi$ is not necessarily compact
is to search for a function $\overline{Z}$ such that for  $Z(\omega):=\overline{Z}(\omega,\langle \omega\rangle)$, $\omega \in \Omega$, one has that $\mathcal{M}(\Xi)=\mathcal{M}_Z(\Xi)$.
While this might not always possible, it turns out that conditions (A1)-(A2) of Assumption~\ref{ass:A} are not far from being a sufficient condition to guarantee that Assumption~\ref{ass:B} is satisfied, provided that one restricts $\ov \Xi$ to contain only H\"older paths. For the precise statement, we refer to the following proposition 
whose proof is presented at the end of Subsection~\ref{subsec:proof.theorem.Z}.
\end{remark}

In the following consider the space of all H\"older continuous functions, 
\[C^{\textup{H\"older}}[0,T]=\bigcup_{n\in\mathbb{N}} \left\{\omega\in C[0,T]:\|\omega\|_{\frac{1}{n}}\le n\right\},\]
where, for $\alpha\in(0,1]$, the $\alpha$-H\"older norm is given by
$\|\omega\|_\alpha:=|\omega(0)|+\sup_{s\neq t}\frac{|\omega(t)-\omega(s)|}{|t-s|^\alpha}$.

\begin{proposition}
\label{prop:sufficient.assumption.for.prediction.set}
  Let $\emptyset\neq \Xi \subseteq \Omega$ be 
  of the form $\Xi=\{\omega \in \Omega\colon (\omega,\langle \omega \rangle) \in \overline \Xi\}$ for some $\ov\Xi \subseteq C[0,T]\times C[0,T]$ such that the conditions (A1)-(A2) of Assumption \ref{ass:A} hold and 
 
  \vspace*{0.15cm}
  \noindent
  $\bullet$  $\ov \Xi  \subseteq C^{\textup{H\"older}}[0,T] \times C^{\textup{H\"older}}[0,T]$,\\
  $\bullet$ $\mathcal{M}(\Xi)$ is nonempty and there exists a constant $c>0$ such that for all $\Q \in \mathcal{M}(\Xi)$
   and $s,t \in [0,T]$ it holds
   \begin{equation}\label{eq:cond:integ:le:M-MZ}
  \mathbb{E}_{\mathbb{Q}}[|S_0|^4]\leq c
   \quad\text{ and } \quad
  \mathbb{E}_{\mathbb{Q}}\big[ |S_t-S_s|^4\big] + \mathbb{E}_{\mathbb{Q}}\big[ |\langle S\rangle_t-\langle S\rangle_s|^4\big] 
   \leq c |t-s|^2.
   \end{equation}
  Then there exists a function $\ov Z\colon C[0,T]\times C[0,T]\to [1,\infty]$ such that Assumption~\ref{ass:B} is satisfied and
    $\mathcal{M}(\Xi)=\mathcal{M}_Z(\Xi)$ for $Z\colon\Omega \to [1,\infty]$ defined by $Z(\omega):=\ov Z(\omega,\langle \omega \rangle)$, $\omega \in \Omega$.
\end{proposition}

\section{Examples}
\label{sec:examples}

\begin{example}[Random $G$-expectation]
	\label{ex:GBM}
	Fix two bounded Borel functions 
  $\underline{\sigma}, \overline{\sigma}\colon [0,T]\times C[0,T]\times C[0,T]\to[0,\infty)$ such that $\underline{\sigma} \leq \overline{\sigma}$ and for every $t\in[0,T]$
  the mappings $(\omega,\nu)\mapsto \underline{\sigma}_t(\omega,\nu)$ and 
  $(\omega,\nu)\mapsto \overline{\sigma}_t(\omega,\nu)$ are continuous. Consider the set
	\[ \overline{\Xi}:=\left\{(\omega,\nu)\in C^{\textup{H\"older}}[0,T] \times C^{\textup{H\"older}}[0,T] :
	\begin{array}{l}
	\omega(0)=\nu(0)=0 \text{ and }\\
  \nu \text{ is nondecreasing with } d\nu\ll d t \text{ and}\\
	d\nu/dt(s) \in[\underline{\sigma}_s(\omega,\nu)^2,\overline{\sigma}_s(\omega,\nu)^2]
  \text{ for all }s
	\end{array}\right\}.\]
	Then $\ov\Xi$ satisfies (A1) and (A2). Indeed, (A2) follows by its definition.
  As for (A1), let 
  $\overline{\Xi}_n:=\overline{\Xi}\cap (C_n\times C^{\textup{H\"older}}[0,T])$
  where $C_n:=\{\omega\in C[0,T] : \|\omega\|_{1/n}\leq n\}$.
  We claim that $\overline{\Xi}_n$ is compact. To that end, observe first that every $(\omega,\nu)\in \ov{\Xi}_n$ satisfies that $\omega \in C^{1/n}[0,T]$ and $\nu$ is Lipschitz-continuous with Lipschitz-constant  $\|\overline{\sigma}\|_\infty$. Therefore, relative compactness of $\ov{\Xi}_n$ follows 
  by the Arzel{\`a}-Ascoli theorem.
  As for closedness 
  let $(\omega^k,\nu^k)_{k \in \N} \subseteq \ov{\Xi}_n$ be a sequence which converges to some point $(\omega,\nu) \in  \ov \Omega$. We need to verify that $(\omega,\nu) \in \ov{\Xi}_n$. Clearly, $\omega(0)=0$ and $\omega\in C^{1/n}[0,T]$, hence it remains to check that $\nu$ satisfies the desired properties so that $(\omega,\nu) \in \ov{\Xi}_n$. Indeed, $\nu(0)=0$, it is nondecreasing, and $\nu \in C^{\textup{H\"older}}[0,T]$ as it is Lipschitz-continuous with constant  $\|\overline{\sigma}\|_\infty$. In particular, $\nu$ is absolutely continuous.
  To check  that $d\nu/dt(s) \in[\underline{\sigma}_s(\omega,\nu)^2,\overline{\sigma}_s(\omega,\nu)^2]
  \text{ for all }s$, observe that as
  $f^k:=d\nu^k/dt$ is bounded from below and above uniformly in $k$, we know from the Banach-Alaoglu theorem that there is a subsequence
  which converges to some $f$ in $\sigma(L^\infty(dt),L^1(dt))$. The continuity of 
  $(\omega,\nu)\mapsto \overline{\sigma}_t(\omega,\nu)$ and 
  $(\omega,\nu)\mapsto \underline{\sigma}_t(\omega,\nu)$ imply that 
  $\underline{\sigma}_t(\omega,\nu)\leq f_t\leq \overline{\sigma}_t(\omega,\nu)$
  $dt$-almost surely. 
  Moreover, as $(\nu_k)_{k \in\N}$ converges to $\nu$ and their derivatives $(f^k)_{k\in \N}$ converge to some $f$ in $\sigma(L^\infty(dt),L^1(dt))$, dominated convergence ensures that $d\nu/dt=f$. Therefore, we obtained that  $(\omega,\nu)\in \ov{\Xi}_n$, which shows that  $\ov{\Xi}_n$ is closed.

  Denote $\Xi=\{\omega\in\Omega : (\omega,\langle\omega\rangle)\in\ov\Xi\}$.
	Using the Burkholder-Davis-Gundy inequality, one obtains for every $\Q\in\mathcal{M}(\Xi)$ and $s,t\in[0,T]$ that
  \begin{align}
  \label{eq:BM.bdg}
  \mathbb{E}_{\mathbb{Q}}[ |S_t-S_s|^4]
  \leq C_{\mathrm{BDG}}\mathbb{E}_{\mathbb{Q}}\Big[ \Big(\int_s^t 1 \,d\langle S\rangle\Big)^2\Big]
  \leq C_{\mathrm{BDG}}\|\ov\sigma\|_\infty^4 |t-s|^2
  \end{align}
  for some constant $C_{\mathrm{BDG}}\geq 0$.
  Further 
  $\mathbb{E}_{\mathbb{Q}}[ |\langle S\rangle_t-\langle S\rangle_s|^4]
  \leq \mathbb{E}_{\mathbb{Q}}[(\int_s^t1\,d\langle S\rangle)^4]
  \leq \|\ov\sigma\|_\infty^8 |t-s|^4
  \leq c\|\ov\sigma\|_\infty^8 |t-s|^2$ for some constant $c\geq 0$ (as $s,t\in[0,T]$ with $T$ finite).
  Therefore Proposition \ref{prop:sufficient.assumption.for.prediction.set} implies
  the existence of $\ov Z$ such that $\Xi$ together with $Z(\omega):=\ov Z(\omega,\langle\omega\rangle)$ for $\omega\in\Omega$ satisfy Assumption \ref{ass:B}.
\end{example}

\begin{remark}
  Observe that in the special case where the functions $\underline{\sigma}$ and $\overline{\sigma}$ are constant we recover the classical $G$-expectation introduced by Peng \cite{Peng.07} when $\underline{\sigma}<\overline{\sigma}$, or the classical Bachelier model when $\underline{\sigma}=\overline{\sigma}$.
\end{remark}

\begin{example}[Black-Scholes under model uncertainty]
\label{ex:black.scholes}
  Fix two bounded Borel functions 
  $\underline{\sigma}, \overline{\sigma}\colon [0,T]\times C[0,T]\times C[0,T]\to[0,\infty)$ such that $\underline{\sigma} \leq \overline{\sigma}$ and for every $t\in[0,T]$
  the mappings $(\omega,\nu)\mapsto \underline{\sigma}_t(\omega,\nu)$ and 
  $(\omega,\nu)\mapsto \overline{\sigma}_t(\omega,\nu)$ are continuous. Consider the set
  \[ \overline{\Xi}:=\left\{(\omega,\nu)\in C^{\textup{H\"older}}[0,T] \times C^{\textup{H\"older}}[0,T] :
  \begin{array}{l}
  \omega(0)=1 \text{ and }
  \nu(\cdot)=\int_0^\cdot \sigma_t^2\omega(t)^2\,dt \text{ where}\\
  \sigma_s^2 \in[\underline{\sigma}_s(\omega,\nu)^2,\overline{\sigma}_s(\omega,\nu)^2]
  \text{ for all }s
  \end{array}\right\}.\]
  Similar arguments as in Example \ref{ex:GBM} show that the assumptions of Propositions \ref{prop:sufficient.assumption.for.prediction.set} are satisfied
  and therefore ensure the existence of $\ov Z$ such that $\Xi$ together with $Z(\omega):=\ov Z(\omega,\langle\omega\rangle)$ for $\omega\in\Omega$ satisfy Assumption \ref{ass:B}.
\end{example}

\begin{remark}
  In case that $\underline{\sigma}=\overline{\sigma}$ is constant, we obtain the classical Black-Scholes model.
\end{remark}


\begin{example}
  Let $c>0$ be a constant, and let 
  \[ \ov\Xi:=\Bigg\{ (\omega,\nu)\in C^{\textup{H\"older}}[0,T] \times C^{\textup{H\"older}}[0,T] :  
  \begin{array}{l}
  \omega(0)=1 \text{ and } \omega\geq 0, \\
  \nu(0)=0 \text{ and } \nu \text{ is nondecreasing},\\
  \tfrac{\sup_{s\in[t,T]} \omega(s) - \inf_{s\in[t,T]} \omega(s)}{T-t}
  \leq c \tfrac{\nu(t)}{t} \,\,\,\forall t\in(0,T)
  \end{array}  
  \Bigg\}.\]
  The motivation of this example is the following. The financial agent believes that at each time, the average future price fluctuation can be controlled by its average volatility observed from the past. 

  To check Assumption \ref{ass:A}, notice first that assumptions (A2)-(A3) follow by definition.
  To verify (A1), for every $n$ define $C^{1/n}[0,T]:=\{\omega\in C[0,T] : \|\omega\|_{1/n}\leq n\}$ and let 
  \[\ov\Xi_n:=\Bigg\{ (\omega,\nu)\in C^{1/n}[0,T] \times C^{1/n}[0,T] :  \begin{array}{l}
  \omega(0)=1 \text{ and } \omega\geq 0, \\
  \nu(0)=0 \text{ and } \nu \text{ is nondecreasing},\\
  \tfrac{\sup_{s\in[t,T]} \omega(s) - \inf_{s\in[t,T]} \omega(s)}{T-t}
  \leq c \tfrac{\nu(t)}{t} \,\,\,\forall t\in\big(\tfrac{1}{n},\tfrac{n-1}{n}\big)
  \end{array}
  \Bigg\}.\]
  As $C^{1/n}[0,T]$ is compact and all constraints are continuous, 
  we see that $\ov\Xi_n$ is a compact set. 
  Since also $\ov\Xi=\bigcup_n\ov\Xi_n$, assumption (A1) and therefore Assumption \ref{ass:A} holds true.
\end{example}
{Finally we provide an example in the spirit of \cite{BartlKupperProemelTangpi.17} showing that if one defines $\Phi(\xi)$ in Theorem~\ref{thm:main.stopped} without the $\liminf$, then there is a duality gap.} 
\begin{example}\label{ex:duality-gap}
{Consider the prediction set
	\[ \Xi:=\Big\{\omega \in C^{\textup{H\"older}}[0,T] :  \langle \omega\rangle \in C^{\textup{H\"older}}[0,T],\, \omega(0)=0, \text{ and }\omega(t)\in[0,1]\text{ for all }t\in[0,T]\Big\}\]
	and the claim $\xi\colon\Omega\to[0,+\infty]$ defined by  $\xi(\omega):=+\infty 1_{\{0\}^c}(\omega)$, meaning that
	\begin{equation*}
	\begin{split}
	\xi(\omega):=
	\begin{cases} 0 & \mbox{if } \omega =0,\\
	+\infty & \mbox{else.}
	\end{cases}
	\end{split}
	\end{equation*}
	One can check that $\Xi$ satisfies  Assumption~\ref{ass:A} and that $\xi$ is lower semicontinuous, in particular satisfies the assumption of Theorem \ref{thm:main.stopped}. 
		Now, observe that the only martingale measure supported on $\Xi$ gives full mass to the constant path 0, hence $$\sup_{\mathbb{Q}\in\mathcal{M}(\Xi)} \mathbb{E}_\mathbb{Q}  [\xi]=\xi(0)=0.$$
		On the other hand, let $\lambda\in[0,+\infty]$ be such that 
		\[\lambda + (H\cdot S)_T+ (G\cdot \mathbb{S})_T\geq \xi\quad\text{on } \Xi\]
		for some $H,G\in\mathcal{H}$ (or, more generally, for some $H$ and $G$ of pathwise finite variation).
		Then for any $\omega\in\Xi\setminus\{0\}$ (e.g.\ $\omega(t)=t/T$) one has $(H\cdot S)_T(\omega)\leq \|t\mapsto H_t(\omega)\|_{1-var} \|\omega\|_\infty<+\infty$ and similarly $(G\cdot \mathbb{S})_T(\omega)<+\infty$, where 
		here $\|\cdot\|_{1-var}$ denotes the first variation norm.
		This requires $\lambda$ to be $+\infty$ and in turn implies that the smallest superhedging price defined with simple (or finite variation) strategies equals $+\infty$, leading to a duality gap.
}
\end{example}
\section{Proof of the main results}
\label{sec:Proof}

\subsection{An enlarged space}
\label{subsec:enlarged space}

All proofs are based on a lifting argument, where instead of $\Omega$ we consider the space 
$\ov \Omega := C[0,T]\times  C[0,T]$; one should think of the first component as the path, and the second
as the it's quadratic variation.
We start by defining all objects for the product space (a basic rule here is that every object with an overline is defined on the product space) and by stating some relationships between the original and the enlarged space.

The space $\ov\Omega$ is endowed with the sup-norm $\|\ov\omega\|_\infty:=\|\omega\|_\infty+\|\nu\|_\infty$ for 
$\ov\omega=(\omega,\nu)\in\ov \Omega$, and equipped with its Borel $\sigma$-field $\ov \cF$.
Denote by
$\fP(\ov \Omega)$ the set of all probability measures on $(\ov \Omega, \ov \cF)$.
The following Borel mapping will play a crucial role
\begin{equation*}
\psi\colon\Omega\to\ov\Omega, 
\quad \quad 
 \omega\mapsto (\omega,\langle\omega\rangle),
\end{equation*}
 and the following set
 \begin{equation*}
 \ov\Delta:=\psi(\Omega)=\big\{(\omega,\nu)\in\ov\Omega: \omega\in\Omega\mbox{ and }\langle\omega\rangle=\nu\big\}.
 \end{equation*}
Moreover, we write $\ov\Delta^c$ for the complement of the set $\ov\Delta$. 
On $\ov\Omega$ we consider the canonical process $(\ov S,\ov V)$ given by $\ov S_t(\ov \omega)=\omega(t)$ and $\ov V_t(\ov\omega)=\nu(t)$ for $\ov\omega=(\omega,\nu)\in\ov \Omega$. 
Also denote $\ov{\S}:=(\ov S^2-\ov{S}^2_0-\ov V)/2$. 
Denote by  $\ov \F =(\ov{\mathcal{F}}_t)_{0\leq t \leq T}$ the raw filtration generated by $(\ov S,\ov V)$ (or, equivalently, by $\ov S$ and $\ov\S$),
by $\ov \F_+$ its right-continuous version, and define the corresponding $\ov\Delta$-augmentations
$\ov\F^{\ov\Delta}:=
\ov\F\vee \sigma(\ov N\subseteq \ov\Omega : \ov N \subseteq \ov\Delta^c)$ and
$\ov\F^{\ov\Delta}_+
:=\ov\F_+\vee \sigma(\ov N\subseteq \ov\Omega : \ov N \subseteq \ov\Delta^c)$.

Further denote by $\ov{\mathcal{H}}$ the set of all $\ov \F^{\ov\Delta}_+$-simple processes
$\ov H\colon [0,T]\times \ov\Omega\to\mathbb{R}$ of the form 
$\ov H_t(\omega)=\sum_{l=1}^{L}  \ov h_l(\ov\omega)1_{( \ov\tau_l({\ov\omega}),\ov\tau_{l+1}({\ov\omega})]}(t)$
for $(t,\ov{\omega})\in [0,T]\times \ov\Omega$,
where $L\in\mathbb{N}$, $0\leq \ov\tau_1\leq\dots\leq \ov\tau_{L+1}\leq T$ are stopping times 
w.r.t.~the filtration $\ov\F_+^{\ov\Delta}$, and $ \ov h_l\colon \ov \Omega\to\mathbb{R}$ 
are bounded $\ov \cF_{\ov \tau_l+}^{\ov\Delta}$-measurable functions. 
For a simple process $\ov H\in\ov{\mathcal{H}}$ the pathwise stochastic integrals
$(\ov H \cdot \ov S)$ and $(\ov H \cdot \ov\S)$ are clearly well-defined as before.

Finally, for any $\ov\Xi\subseteq\ov\Omega$  and $\ov Z\colon \ov \Omega \to [1,\infty]$, let
\begin{equation*}
\begin{split}
\ov{\cM}(\ov \Xi)&:=\big\{ \ov \Q \in \fP (\ov \Omega)\colon \ov S \mbox{ and } \ov \S \mbox{ are $\ov \Q$-$\ov \F$-local martingales and }  \ov \Q(\ov \Xi)=1 \big\},\\
\ov{\cM}_{\ov Z}(\ov \Xi)&:=\big\{ \ov \Q \in \fP (\ov \Omega)\colon \ov S \mbox{ and } \ov \S \mbox{ are $\ov \Q$-$\ov \F$-local martingales, }\, \ov \Q(\ov \Xi)=1, \text{ and } \E_{\ov \Q}[\ov Z]<\infty \big\}.
\end{split}
\end{equation*}

\begin{remark}\label{rem:filtration-representation}
  An easy computation (similar as for the completion of a $\sigma$-field) shows that
  $\ov\cF^{\ov\Delta}_t=\{(\ov A\cap\ov\Delta)\cup \ov N : \ov A\in\ov \cF_t,\, \ov N\subseteq \ov\Delta^c\}$,
  and
  $\ov\cF^{\ov \Delta}_{t+}=\{(\ov A\cap\ov\Delta)\cup \ov N : \ov A\in\ov \cF_{t+},\, \ov N\subseteq \ov\Delta^c\}$.
\end{remark}
\begin{remark}\label{rem:ass-A2-enlarged}
Notice that if $\ov \Xi$ satisfies condition~(A2) of Assumption~\ref{ass:A}, then  standard results on the quadratic variation ensure that for every
$\ov{\Q}\in\ov{\mathcal{M}}(\ov\Xi)$ we have that $\langle \ov S\rangle^{\ov{\Q}}=\ov V$ $\ov{\Q}$-a.s.\ and 
  $\ov\Q(\ov S\in \Omega)=1$ by Remark \ref{rem:Omega},
  which shows $\ov{\Q}(\ov\Delta)=1$. Hence, $\ov{\Q}(\ov N)=0$ for all $\ov N\subseteq\ov\Delta^c$, and hence $\ov\F^{\ov\Delta}\subseteq \ov\F^{\ov{\Q}}$ and 
  $\ov\F^{\ov\Delta}_+\subseteq \ov\F^{\ov{\Q}}_+$, where e.g.~$\ov\F^{\ov{\Q}}$ denotes the $\ov{\Q}$-completion of $\ov\F$. 
\end{remark}

\begin{lemma}
\label{lem:transfer.adapted}
  For every $t\in[0,T]$ one has
\begin{itemize}
\item[\emph{(a)}] $\ov\cF^{\ov \Delta}_{t+}=\bigcap_{s>t} \ov\cF^{\ov\Delta}_{s\wedge T}$,
  
\item[\emph{(b)}] $\{\ov\omega\in\ov\Omega : \omega\in A\}\in\ov{\mathcal{F}}_{t}^{\ov\Delta}$
    for all $A\in \mathcal{F}_{t}$, 

\item[\emph{(c)}] $\{\ov\omega\in\ov\Omega : \omega\in A\}\in\ov{\mathcal{F}}_{t+}^{\ov\Delta}$
    for all $A\in \mathcal{F}_{t+}$, 
    
\item[\emph{(d)}] $\psi^{-1}(\ov A)\in\mathcal{F}_{t}$ for all $\ov{A}\in\ov{\mathcal{F}}_{t}^{\ov\Delta}$,

\item[\emph{(e)}] $\psi^{-1}(\ov A)\in\mathcal{F}_{t+}$ for all $\ov{A}\in\ov{\mathcal{F}}_{t+}^{\ov\Delta}$.
\end{itemize}
\end{lemma}
\begin{proof} 
  (a) By definition, we have $\ov\cF^{\ov \Delta}_{t+}\subseteq \bigcap_{s>t} \ov\cF^{\ov\Delta}_{s\wedge T}$. For the reverse inclusion, fix $t\in[0,T]$ and $\ov B \in \bigcap_{s>t} \ov\cF^{\ov\Delta}_{s\wedge T}$. Then, for every $s>t$, there exist $\ov A_s \in \ov \cF_s$ and $\ov N_s \subseteq \ov \Delta^c$ such that
  \begin{equation*}
  \ov B = (\ov B \cap \ov \Delta) \cup (\ov B \cap \ov \Delta^c)= (\ov A_s \cap \ov \Delta) \cup \ov N_s
  \end{equation*}
  and both unions are disjoint.
  This implies for each $s>t$ that $\ov A_s \cap \ov \Delta=\ov B \cap \ov \Delta$, as well as $\ov N_s=\ov B \cap \ov \Delta^c$. Define
  \begin{equation*}
  \ov A:= \bigcap_{s>t,\, s\in \Q}\ \bigcup_{t<r<s,\, r\in \Q} \ov A_r \quad \in \ov \cF_{t+}
  \end{equation*}
  Then by construction $\ov B = (\ov A \cap \ov \Delta) \cup (\ov B \cap \ov \Delta^c)$ where the union is disjoint,
  proving the reverse inclusion.
  
  (b) Fix $t\in[0,T]$ and $A\in \mathcal{F}_t$, so that
  \[ A=\big\{\omega\in\Omega : (\omega(s))_{s\in S}\in B\big\}\]
  for a countable set $S\subseteq[0,t]$ and a Borel set $B\subseteq \mathbb{R}^S$.
  Let $\ov N:=\{\ov\omega\in \ov{\Delta}^c : \omega\in A\}\in\ov{\mathcal{F}}_t^{\ov\Delta}$.
  Since $\ov\omega\mapsto \omega(s)$ is $\ov{\mathcal{F}}_{t}$-measurable 
for all $s\leq t$, one has
  \begin{align*}
  \big\{\ov\omega\in\ov{\Omega} : \omega\in A\big\}
  =\big\{\ov\omega\in\ov{\Delta} : \omega\in A\big\}\cup \ov N
  &=\big\{ \ov\omega\in \ov{\Delta} : (\omega(s))_{s\in S} \in B \big\}\cup \ov N
  \in\ov{\mathcal{F}}_{t}^{\ov\Delta}.
  \end{align*}
  
(c) If  $A\in \mathcal{F}_{t+}$ then $A\in\mathcal{F}_{s\wedge T}$ for every $s>t$, so that by (a)
\[\{\ov\omega\in\ov\Omega : \omega\in A\}\in\bigcap_{s>t} \ov\cF^{\ov\Delta}_{s\wedge T}=\ov{\mathcal{F}}_{t+}^{\ov\Delta}.\]

(d) Fix $t\in[0,T]$ and $\ov{A}\in\ov{\mathcal{F}}_t^{\ov\Delta}$.
  Then 
  \[\ov A=\big\{\ov\omega\in\ov\Delta: (\omega(r),\nu(s))_{(r,s)\in S}\in B\big\}\cup \ov N\]
  for a countable set $S\subseteq[0,t]\times [0,t]$, a
  Borel set $B\subseteq\mathbb{R}^S$,
  and $\ov N\subseteq \ov{\Delta}^c$.
  Since $\psi(\omega)\in\ov \Delta$ for every $\omega\in\Omega$, it follows that
  \[ \psi^{-1}(\ov A)
  =\big\{\omega\in\Omega : (\omega(r),\langle\omega\rangle(s))_{(r,s)\in S}\in B \big\}
  \in\mathcal{F}_{t},\]
  since $\omega\mapsto (\omega(r),\langle\omega\rangle(s))$ is $\mathcal{F}_t$-measurable
  for every $r,s\leq t$. 

(e) The argumentation is similar to (c).
\end{proof}
\begin{remark}\label{rem:stopping}
	The following simple results which we will use frequently are direct consequences of Lemma~\ref{lem:transfer.adapted}.
	Let $h:\Omega\to\R$ be $\mathcal{F}_{t+}$-measurable and $\tau$ a $\F_{+}$-stopping time.
	Then $\ov h(\ov\omega):=h(\omega) 1_{\ov\Delta}(\ov\omega)$, $\ov \omega:=(\omega,\nu) \in \ov\Omega$, 
	is $\ov{\mathcal{F}}^{\ov \Delta}_{t+}$-measurable and $\ov\tau(\ov\omega):=\tau(\omega)\,1_{\ov\Delta}(\ov\omega)+T\,1_{\ov\Delta^c}(\ov\omega)$, $\ov\omega:=(\omega,\nu)  \in \ov\Omega$, is a $\ov\F^{\ov \Delta}_+$-stopping time. Conversely, let $\ov\tau$ be a $\ov\F^{\ov \Delta}_+$-stopping time. Then $\ov \tau \circ \psi$ is a $\F_{+}$-stopping time.
\end{remark}
The next lemmas are the key tools to go back and forth between the original space $\Omega$ and the enlarged space $\ov \Omega$.
\begin{lemma}
  \label{lem:transfer.integrals}
  For every $H,G\in\mathcal{H}$ there are $\ov H, \ov G\in\ov{\mathcal{H}}$ such that
  \[ (\ov H\cdot \ov S)+ (\ov G\cdot \ov{\mathbb{S}})
  = ( H\cdot S)\circ \ov S + (G\cdot \mathbb{S})\circ \ov S
  \quad\text{on }\ov\Delta.\]
  Conversely, for every $\ov H,\ov G\in\ov{\mathcal{H}}$ there are $H,G\in \mathcal{H}$ such that
  \[(H\cdot S) + (G \cdot \mathbb{S})
    = (\ov H \cdot \ov S) \circ\psi + (\ov G \cdot \ov{\mathbb{S}})\circ\psi
  \quad\text{on }\Omega.\]
\end{lemma}
\begin{proof}
  Let $H=\sum_{l=1}^{L}  h_l \,1_{(\tau_l,\tau_{l+1}]}$ and $G$ be elements of $\mathcal{H}$. 
  Define $\ov H:=\sum_{l=1}^{L}  \ov{h}_l \,1_{(\ov{\tau}_l,\ov{\tau}_{l+1}]}$,
  where
  \[ \ov{h}_l(\ov\omega):=h_l(\omega)\,1_{\ov\Delta}(\ov\omega)
  \quad\text{and}\quad
  \ov{\tau}_l(\ov\omega):=\tau_l(\omega)\,1_{\ov\Delta}(\ov\omega)
  +T\,1_{\ov{\Delta}^c}(\ov\omega)\]
  for all $l$ and $\ov\omega\in\ov\Omega$. Define $\ov G$ analogously.
  By Lemma \ref{lem:transfer.adapted} one has $\ov{H},\ov G\in\ov{\mathcal{H}}$.
  Then, since $\ov{\mathbb{S}}=\mathbb{S}\circ\ov S$ on $\ov\Delta$, 
  \[ (\ov H\cdot \ov S)+ (\ov G\cdot \ov{\mathbb{S}})
  = ( H\cdot S)\circ \ov S + (G\cdot \mathbb{S})\circ \ov S
  \quad\text{on }\ov\Delta.\]
  
  To prove the opposite inequality, let 
  $\ov H=\sum_{l=1}^{L}  \ov{h}_l \,1_{(\ov{\tau}_l,\ov{\tau}_{l+1}]}$ and $\ov G$ be elements of $\ov{\mathcal{H}}$.
  Define $H:=\sum_{l=1}^{L}  (\ov{h}_l\circ\psi) \, 1_{(\ov\tau_l\circ\psi,\ov\tau_{l+1}\circ\psi]}$
  and $G$ analogously which are both elements of $\mathcal{H}$ by Lemma \ref{lem:transfer.adapted}.
  Since  
  \[ \mathbb{S}=\ov{\mathbb{S}}\circ\psi,
  \quad \psi(\Omega) = \ov\Delta, 
  \quad\text{and}\quad \ov S \circ \psi 
  =\mathrm{\mathop{id}}_\Omega \]
  it holds
  \begin{align*}
    &(H\cdot S) + (G \cdot \mathbb{S})
    = ((\ov H\circ \psi) \cdot S + ((\ov G \circ \psi) \cdot \mathbb{S})\\
    &= ((\ov H \circ \psi ) \cdot (\ov S \circ\psi))   + ((\ov G \circ \psi) \cdot (\ov{\mathbb{S}} \circ\psi))
    = (\ov H \cdot \ov S) \circ\psi + (\ov G \cdot \ov{\mathbb{S}})\circ\psi
    \quad\text{on } \Omega
  \end{align*}
  which proves the claim.
\end{proof}

\begin{lemma}\label{lem:M-M-large-trans}
Let $\Xi\subseteq \Omega$ be of the form $\Xi=\{\omega \in \Omega\colon (\omega,\langle \omega \rangle) \in \overline \Xi\}$ for some $\ov\Xi \subseteq C[0,T]\times C[0,T]$ satisfying condition (A2) of Assumption~\ref{ass:A}. Then
\label{lem:transfer.measures}
  \[\cM(\Xi)=\big\{ \ov \Q \circ \ov S^{-1}: \ \ov \Q \in \ov \cM(\ov\Xi)\big\}
  \quad\text{and}\quad
  \ov \cM(\ov \Xi)=\big\{ \Q \circ \psi^{-1}: \  \Q \in \cM(\Xi)\big\}.
  \]
  Moreover, let additionally $\ov Z\colon \Omega \to [1,\infty]$ be a function and $Z:=\ov Z \circ \psi$.
  Then 
   \[\cM_{Z}(\Xi)=\big\{ \ov \Q \circ \ov S^{-1}: \ \ov \Q \in \ov \cM_{\ov Z}(\ov\Xi)\big\}
  \quad\text{and}\quad
  \ov \cM_{\ov Z}(\ov \Xi)=\big\{ \Q \circ \psi^{-1}: \  \Q \in \cM_Z(\Xi)\big\}.
  \]
  same holds true for $\cM(\Xi)$ and $\ov \cM(\ov\Xi)$ replaced by 
  $\cM_Z(\Xi)$ and $\ov \cM_{\ov Z}(\ov\Xi)$, respectively.
\end{lemma}
\begin{proof}
  Let first $\mathbb{Q} \in \mathcal{M}(\Xi)$ and define $\ov \Q := \Q \circ \psi^{-1}$.
  As $\Xi=\psi^{-1}(\ov\Xi)$,  we have that $\ov\Q(\ov\Xi)=\Q(\Xi)=1$.
  Moreover, to see that $\ov S$ is a $\ov \Q$-$\ov \F$-local martingale, define for each $m \in \N$ 
  \begin{equation}\label{eq:loc-seq-orig}
  \tau_m:= \inf\{t\geq 0 \colon |S_t|\geq m\}.
  \end{equation}
  Since $S$ has continuous paths, 
  $(\tau_m)_{m \in \N}$ is a $\Q$-$\F$-localizing sequence for $S$, i.e.\ a sequence of $\F$-stopping times with $\lim_{m\to \infty} \tau_m=\infty \  \Q$-a.s.\ such that $S^{\tau_m}$ is a $\Q$-$\F$-martingale for each $m \in \N$.  Define for each $m$
 \begin{equation}\label{eq:loc-seq-enlar}
 \ov \tau_m:= \inf\{t\geq 0 \colon |\ov S_t|\geq m\}.
 \end{equation}
  Then  we get that $(\ov \tau_m)_{m \in \N}$ is a $\ov\Q$-$\ov \F$-localizing for $\ov S$, since for every $m \in \N$, $0\leq s\leq t\leq T$ and $\ov A\in \ov{\mathcal{F}}_s$
  \[ \E_{\ov\Q}[ ({\ov S}^{\ov{\tau}_m}_t-{\ov S}^{\ov{\tau}_m}_s)1_{\ov A}]
  =\E_{\Q}[ (S^{{\tau}_m}_t-S^{{\tau}_m}_s)1_{\psi^{-1}(\ov A)}]
  =0 \]
  by Lemma \ref{lem:transfer.adapted} and the martingale property of $S^{\tau_m}$ under $\Q$.
  Similarly, since $\psi(S)=(S,\langle S\rangle^{\mathbb{Q}})$ $\mathbb{Q}$-almost surely,
  it follows by the same argument that $\ov{\mathbb{S}}$ is a $\ov \Q$-$\ov \F$-local martingale.
  Therefore, we get that indeed $\ov\Q\in \ov{\mathcal{M}}(\ov\Xi)$.  Furthermore, if in fact $\Q \in \cM_Z(\Xi)$, then 
  \begin{equation*}
  \E_{\ov\Q}[\ov Z]= \E_{\Q}[\ov Z \circ \psi] = \E_{\Q}[Z]<\infty
  \end{equation*}
  ensures that  also $\ov \Q \in \ov{\cM}_{\ov Z}(\ov\Xi)$.

  For the other direction, let  $\ov  \Q \in \ov{\mathcal{M}}(\ov \Xi)$.
  Recall that due to condition (A2) we have $\ov \Q$-a.s.~that $\ov V$ is nondecreasing with $\ov V_0=0$. Therefore since both  $\ov S$ and $2\ov{\mathbb{S}}= \ov S^2 -\ov{S}^2_0 -\ov V$ are $\ov \Q$-$\ov \F$-local martingales, one concludes that $\ov V =\langle \ov S \rangle^{\ov \Q} \ \ov \Q$-a.s., and thus
  $\psi(\ov S)=(\ov S,\langle \ov S \rangle^{\ov \Q})$ $\ov \Q$-a.s..
  In particular $\ov\Q(\ov S\in \Omega)=1$, so that
  $\Q:= \ov \Q \circ \ov S^{-1}$ 
  defines a Borel probability measure on $\Omega$.
  Moreover, we know from Remark~\ref{rem:ass-A2-enlarged} that $\ov \Q(\ov \Delta)=1$ and as $\ov S^{-1}(\Xi)\cap\ov\Delta=\ov \Xi\cap\ov\Delta$, one has $\Q(\Xi)=\ov\Q(\ov\Xi)=1$.
  Finally, using \eqref{eq:loc-seq-orig} and \eqref{eq:loc-seq-enlar}, we deduce for all $m \in \N$, $0\leq s\leq t\leq T$, and $A\in\mathcal{F}_s$
   from  Lemma \ref{lem:transfer.adapted} and $\ov \cF^{\ov \Delta}_s\subseteq \ov \cF^{\ov\Q}_s$ that
  \[ \E_{\Q}[ (S^{\tau_m}_t-S^{\tau_m}_s)1_A]
  =\E_{\ov\Q}[ (\ov S^{\ov{\tau}_m}_t-\ov S^{\ov{\tau}_m}_s)1_{\{\ov\omega\in\ov\Omega \,:\, \omega\in A\}}]
  =0, \]
  which shows that $S$ is a $\Q$-$\F$-local martingale. Thus $\Q\in\mathcal{M}(\Xi)$. 
  Furthermore, if in fact $\ov \Q \in \ov{\cM}_{\ov Z}(\ov \Xi)$, then since $\psi \circ \ov S= \mathrm{\mathop{id}}_{\ov\Omega}$ on $\ov \Delta$ and since $\ov \Q(\ov \Delta)=1$ due to condition (A2), we have that
  \begin{equation*}
  \E_{\Q}[Z]
  = \E_{\Q}[\ov Z \circ \psi] 
  = \E_{\ov \Q}[\ov Z \circ \psi \circ \ov S ]
  = \E_{\ov \Q}[\ov Z]
  <\infty
  \end{equation*}
  which implies that then $\Q  \in \cM_{Z}(\Xi)$.
\end{proof}

\subsection{Proof of Theorem \ref{thm:main.stopped}}\label{subsec:proof-main-stopped}
In this subsection we provide the proof of Theorem~\ref{thm:main.stopped}. To that end, we fix a prediction set $\Xi\subseteq \Omega$ which satisfies Assumption~\ref{ass:A}.

We start by providing two transition lemmas between the original space $\Omega$ and the enlarged space  $\ov \Omega$ which are in fact direct consequences of the crucial Lemmas \ref{lem:transfer.integrals} and \ref{lem:M-M-large-trans}.  Denote by $C_b(\ov\Omega)$ and $U_b(\ov\Omega)$ the spaces of bounded continuous
and bounded upper semicontinuous functions $\ov \xi\colon \ov\Omega\to\R$, respectively. 
We start with the primal problem and define for every function $\ov\xi\colon \ov \Omega \to (-\infty, \infty]$ the functional
\begin{equation}\label{eq:def:Phi-enlarged}
  \ov \Phi(\ov \xi):=\inf\left\{ \lambda\in\mathbb{R} : 
\begin{array}{l}
\text{there are sequences } (\ov H^n),(\ov G^n)\in\ov{\mathcal{H}} \text{ such that}\\
\lambda +(\ov H^n\cdot \ov S)_t + (\ov G^n\cdot \ov{\mathbb{S}})_t\geq 0 \text{ on } \ov \Delta \cap \ov \Xi \text{ for all } n, t, \\
\lambda+\liminf_n( (\ov H^n\cdot \ov S)_T + (\ov G^n\cdot \ov{\mathbb{S}})_T)\geq \ov \xi \mbox{ on } \ov\Delta\cap\ov\Xi
\end{array}\!\right\}
\end{equation}
We have the following result.
\begin{lemma}\label{lem:primal-transition}
Let $\Phi$ be the functional defined in Theorem~\ref{thm:main.stopped} and $\ov \Phi$ be the functional defined in \eqref{eq:def:Phi-enlarged}. Then for every $\xi\colon C[0,T] \to (-\infty,\infty]$ we have that $\Phi(\xi)= \ov \Phi (\xi \circ \ov S)$.
\end{lemma}
\begin{proof}
To see the first inequality $\Phi(\xi)\geq \ov \Phi(\xi \circ  \ov S)$, let $\lambda>\Phi(\xi)$. 
By definition, there exist sequences  $(H^n), (G^n)$ in $\mathcal{H}$ such that 
\begin{align*}
&\lambda + (H^n \cdot S)_t + (G^n \cdot \mathbb S)_t \geq 0 \text{ on } \Xi \text{ for all } n,t,\\
&\text{and }\lambda + \liminf_{n \to \infty}\big((H^n\cdot S)_T + (G^n \cdot \mathbb{S})_T\big)\geq \xi \text{ on } \Xi.
\end{align*}
By Lemma~\ref{lem:transfer.integrals} we know that for each $n \in \N$ there are $\ov H^n, \ov G^n\in\ov{\mathcal{H}}$ such that
\[ (\ov H^n\cdot \ov S)+ (\ov G^n\cdot \ov{\mathbb{S}})
= ( H^n\cdot S)\circ \ov S + (G^n\cdot \mathbb{S})\circ \ov S
\quad\text{on }\ov\Delta.\]
Therefore, as $\ov S(\ov \Delta \cap \ov \Xi)=\Xi$, one has
\begin{align*}
&\lambda + (\ov H^n\cdot \ov S)_t+(\ov G^n\cdot \ov{\mathbb{S}})_t\geq 0
\text{ on } \ov\Delta\cap\ov\Xi \text{ for all }n,t\,\\
& \mbox{and }\lambda+ \liminf_{n \to \infty}\big( (\ov H^n\cdot \ov S)_T +(\ov G^n\cdot \ov{\mathbb{S}})_T \big)
\geq  \xi \circ \ov S   \text{ on } \ov\Delta\cap\ov\Xi.
\end{align*}
As $\lambda>\Phi(\xi)$ was arbitrary, the inequality $\Phi(\xi)\geq \ov \Phi(\xi\circ \ov S)$ follows.

To prove the opposite inequality, let $\lambda>\ov \Phi(\xi \circ \ov S)$. 
Then, by definition, there exist sequences $(\ov H^n)$, $(\ov G^n)$ in $\ov \cH$ such that 
\begin{align*}
&\lambda + (\ov H^n\cdot \ov S)_t+(\ov G^n\cdot \ov{\mathbb{S}})_t\geq 0
\text{ on } \ov\Delta\cap\ov\Xi \text{ for all }n,t\,\\
&\mbox{and }\lambda+ \liminf_{n \to \infty}\big( (\ov H^n\cdot \ov S)_T +(\ov G^n\cdot \ov{\mathbb{S}})_T \big)\geq \xi \circ \ov S 
\quad \text{ on } \ov\Delta\cap\ov\Xi.
\end{align*}
By Lemma~\ref{lem:transfer.integrals} we know that for each $n \in \N$ there are
$H^n,G^n\in \mathcal{H}$ such that
\[(H^n\cdot S) + (G^n \cdot \mathbb{S})
= (\ov H^n \cdot \ov S) \circ\psi + (\ov G^n \cdot \ov{\mathbb{S}})\circ\psi
\quad\text{on }\Omega.\]
 Therefore, as $\psi(\Xi)=\ov \Delta \cap \ov \Xi$, it follows that
\begin{align*}
&\lambda + (H^n \cdot S)_t + (G^n \cdot \mathbb{S})_t
\geq 0 \text{ on } \Xi \text{ for all }n,\,t\\
& \mbox{and }
\lambda + \liminf_{n \to \infty} \big((H^n \cdot S)_T + (G^n \cdot \mathbb{S})_T \big) 
\geq  \ \xi\circ \ov S \circ \psi
=   \xi
\quad\text{on }\Xi.
\end{align*}
As $\lambda>\ov\Phi(\xi\circ\ov S)$ was arbitrary, the inequality $\ov \Phi(\xi\circ \ov S)\geq \Phi(\xi)$ follows.
\end{proof}
For the dual problem, we have the following transition lemma
\begin{lemma}\label{lem:dual-transition}
Let $\Xi\subseteq \Omega$ be of the form $\Xi=\{\omega \in \Omega\colon (\omega,\langle \omega \rangle) \in \overline \Xi\}$ for some $\ov\Xi \subseteq C[0,T]\times C[0,T]$ satisfying condition (A2) of Assumption~\ref{ass:A}. Then for every Borel function $\xi\colon C[0,T] \to (-\infty,\infty]$ which is bounded from below we have that
\begin{equation*}
\sup_{\Q \in \cM(\Xi)} \E_{\Q}[\xi]=\sup_{\ov{\Q} \in \ov{\cM}(\ov\Xi)} \E_{\ov\Q}[\xi\circ \ov S].
\end{equation*}
\end{lemma}
\begin{proof}
To see the first inequality 
$\sup_{\Q \in \cM(\Xi)} \E_{\Q}[\xi]\leq \sup_{\ov{\Q} \in \ov{\cM}(\ov\Xi)} \E_{\ov\Q}[\xi\circ \ov S]$, let $\Q \in \cM(\Xi)$. Then, by Lemma~\ref{lem:M-M-large-trans} we know that $\ov \Q:= \Q \circ \psi^{-1} \in \ov{\cM}(\ov \Xi)$.
Therefore, we see that
\begin{equation*}
\E_{\Q}[\xi]
= \E_{\Q}[\xi\circ\ov S \circ \psi]
=\E_{\ov \Q}[\xi\circ\ov S]
\leq \sup_{\ov{\Q} \in \ov{\cM}(\ov\Xi)} \E_{\ov\Q}[\xi\circ \ov S].
\end{equation*}
As $\Q \in \cM(\Xi)$ was arbitrary, we obtain indeed the first inequality.

For the reverse inequality $\sup_{\Q \in \cM(\Xi)} \E_{\Q}[\xi]\geq \sup_{\ov{\Q} \in \ov{\cM}(\ov\Xi)} \E_{\ov\Q}[\xi\circ \ov S]$,
let $\ov{\Q} \in \ov{\cM}(\ov\Xi)$. Then by Lemma~\ref{lem:M-M-large-trans} we know that $\Q:= \ov \Q\circ \ov S^{-1} \in \cM (\Xi)$.
Therefore,  we get that
\begin{equation*}
\E_{\ov \Q}[\xi \circ \ov S]= \E_{\Q}[\xi] \leq \sup_{\Q \in \cM(\Xi)} \E_{\Q}[\xi].
\end{equation*} 
Since $\ov{\Q} \in \ov{\cM}(\ov\Xi)$ was arbitrary, we also obtain the reverse inequality. 
\end{proof}
Having the transition lemmas in mind, it remains to prove our result in Theorem~\ref{thm:main.stopped} on the enlarged space.
\begin{proposition}
\label{prop:lifted.dual.with.stopped.paths}
Let $\Xi\subseteq \Omega$ be a prediction set satisfying Assumption~\ref{ass:A} and let $\ov\Phi$ be the functional defined in \eqref{eq:def:Phi-enlarged}.
  Then
  \begin{align*}
  \ov \Phi(\ov \xi)
  =\sup_{\ov{\mathbb{Q}} \in\ov{\mathcal{M}}(\ov\Xi)} \mathbb{E}_{\ov{\mathbb{Q} }}[\ov \xi]
  \end{align*}
  for every $\ov\xi\colon\ov\Omega\to[0,\infty]$ which can be written as 
  $\ov\xi=\liminf_n\ov\xi_n$ where $\ov\xi_n\colon\ov\Omega\to[0,\infty)$ are bounded
  upper-semicontinuous functions.
\end{proposition}
\begin{proof}
  The proof is divided into three main steps. First,
  also note that by Lemma~\ref{lem:stopped.prediction.set} there is an
  increasing sequence of nonempty compact sets 
  $\ov{\Xi}_n\subseteq C[0,T]\times C[0,T]$, $n \in \N$,   such that $\ov\Xi=\bigcup_n\ov\Xi_n$ 
  and for all $n$ we have that $\ov\omega \in \ov \Xi_n$ implies that $\ov{\omega}^t\in\ov{\Xi}_n$ for every $t$. As a consequence, as $\ov\Xi_n$ is nonempty and contains at least one constant path,  $\ov{\mathcal{M}}(\ov\Xi_n)$ is nonempty, too, as it
  contains at least one constant martingale measure.
   
  Step (a): Fix $n\in\mathbb{N}$. For any Borel function $\ov \xi \colon \ov\Omega\to(-\infty,\infty]$ which is bounded from below, define
  \begin{align*}
  \ov \Phi_n(\ov \xi)&:=\inf\left\{ \lambda\in\mathbb{R} : 
       \begin{array}{l}
       \text{there are $\ov H,\ov G\in\ov{\mathcal{H}}$ such that}\\
       \text{$\lambda+(\ov H\cdot \ov S)_T + (\ov G\cdot \ov{\mathbb{S}})_T\geq \ov \xi \mbox{ on } \ov\Delta\cap\ov\Xi_n$}
       \end{array} \right\}.
  \end{align*}
  Then, we claim that for all $\ov{\Q}\in\ov{\mathcal{M}}(\ov{\Xi}_n)$ 
  \begin{align}
  \label{eq:weakduality} 
  \ov \Phi_n(\ov \xi)&\geq \E_{\ov{\Q}}[\ov \xi]  
  \end{align}
  as well as for all $\ov{\Q}\in\ov{\mathcal{M}}(\ov{\Xi})$ that
  \begin{equation}\label{eq:weakduality-part2} 
  \ov \Phi(\ov \xi)\geq \E_{\ov{\Q}}[\ov \xi].
  \end{equation}
  In particular, the functional $\ov \Phi_n$ is real-valued on $U_b$.
  
  Indeed, let $\ov\Q \in \ov{\cM}(\ov \Xi_n)$ and
  let $\lambda >\ov \Phi_n(\ov \xi)$ so that there exist
  $\ov H,\ov G\in\ov{\mathcal{H}}$ such that
  $\lambda+(\ov H\cdot \ov S)_T + (\ov G\cdot \ov{\mathbb{S}})_T\geq \ov \xi$  on $\ov \Delta\cap\ov\Xi_n$. 
  Notice that as $\ov\Xi_n$ is compact, every local martingale with respect to $\ov\Q \in \ov{\cM}(\ov \Xi_n)$ is in fact a true martingale.
  Therefore, as by Remark~\ref{rem:ass-A2-enlarged} both $\ov S$ and $\ov{\mathbb{S}}$ are $\ov \Q$-$\ov \F^{\ov\Delta}_+$-martingales and $\ov H$, $\ov G$ are simple integrands, it follows that 
  $\E_{\ov \Q}[(\ov H \cdot \ov S)_T]=\E_{\ov \Q}[(\ov G \cdot \ov{\mathbb{S}})_T]=0$. Hence
  $\lambda \geq \E_{\ov{\Q}}[\ov \xi]$ which shows  \eqref{eq:weakduality}.
  To see \eqref{eq:weakduality-part2}, let $\ov\Q \in \ov{\cM}(\ov \Xi)$  and 
  $\lambda >\ov \Phi(\ov \xi)$ so that there exist sequences
  $(\ov H^k), (\ov G^k)\in\ov{\mathcal{H}}$ such that 
  $\lambda +(\ov H^k\cdot \ov S)_t + (\ov G^k\cdot \ov{\mathbb{S}})_t\geq 0 \text{ on } \ov \Delta \cap \ov \Xi \text{ for all } k, t,$ and 
  $\lambda+\liminf_k( (\ov H^k\cdot \ov S)_T + (\ov G^k\cdot \ov{\mathbb{S}})_T\geq \ov \xi \mbox{ on } \ov\Delta\cap\ov\Xi$. The admissibility condition ensures that for each $k$, the process $\ov M^k:=(\ov H^k\cdot \ov S) + (\ov G^k\cdot \ov{\mathbb{S}})$ is a $\ov\Q$-$\ov \F^{\ov\Delta}_+$-supermartingale starting in zero. 
  Therefore, applying Fatou's lemma yields 
  \begin{equation*}
  \lambda\geq \lambda + \liminf_{k \to \infty} \E_{\ov \Q}\big[\ov M^k_T\big]
  \geq  \E_{\ov \Q}\big[\lambda + \liminf_{k \to \infty}\ov M^k_T\big] 
  \geq \E_{\ov \Q}[\ov\xi],
  \end{equation*}
  which in turn implies \eqref{eq:weakduality-part2}.

  Step (b): Each $\ov \Phi_n$ is continuous from above on $C_b(\ov\Omega)$, that is,
  for every sequence $(\ov \xi_k)$ in $C_b(\ov\Omega)$ which decreases pointwise to $0$, one
  has $\ov\Phi_n(\ov \xi_k)\downarrow \ov \Phi_n(0)$.
  This follows from Dini's lemma as $\ov\Xi_n$ is compact.
  The non-linear Daniell-Stone theorem \cite[Theorem 2.2]{BartlCheriditioKupper.17} therefore implies
  \begin{align}
  \label{eq:rep.phi.n.on.enlarged.space}
  \ov \Phi_n(\ov \xi)
  =\sup_{\ov\Q\in\ov{\mathcal{M}}(\ov\Xi_n) } \ov\E_{\ov\Q}[\ov\xi]
  \quad\text{for all }\ov\xi\in U_b(\ov\Omega),
  \end{align}
  provided that we can show that for every finite Borel measure $\ov\Q$ on $\ov\Omega$ one has
  \begin{align}
  \begin{split}
  \label{eq:conjugate}
  \ov \Phi_n^\ast(\ov{\Q})
    :=\sup_{\ov \gamma\in C_b(\ov\Omega)} (\E_{\ov{\Q}}[\ov \gamma]-\ov \Phi_n(\ov \gamma))
    &=\sup_{\ov \gamma\in U_b(\ov\Omega)} (\E_{\ov{\Q}}[\ov \gamma]-\ov\Phi_n(\ov \gamma))
    =\begin{cases}
    0, &\text{if } \ov{\Q}\in\ov{\mathcal{M}}(\ov \Xi_n),\\
    \infty,&\text{else}.
    \end{cases}
  \end{split}
  \end{align}

  To that end, let first $\ov{\Q}\in\ov{\mathcal{M}}(\ov \Xi_n)$. Then   by \eqref{eq:weakduality} we have for every
  $\ov \gamma\in U_b(\ov\Omega)$  that
  $\E_{\ov{\Q}}[\ov \gamma]-\ov \Phi_n(\ov \gamma)\leq 0$. 
  In particular, 
  \begin{equation}\label{eq:Step-c-easy-dir}
  \ov \Phi_n^\ast(\ov{\Q})
  \leq \sup_{\ov \gamma\in U_b(\ov\Omega)} (\E_{\ov{\Q}}[\ov \gamma]-\ov\Phi_n(\ov \gamma))
  \leq 0.
  \end{equation}
   In addition, for every $m\in\mathbb{R}$ it holds that $\ov \Phi_n(m)\leq m$ by definition. Therefore, we conclude that
   \begin{equation}\label{eq:Step-c-easy-dir-Part2}
   \ov\Phi_n^\ast(\ov{\Q})\geq \E_{\ov \Q}[0]-\ov\Phi_n(0)\geq 0.
   \end{equation}
   This together with \eqref{eq:Step-c-easy-dir} show that indeed for every $\ov{\Q}\in\ov{\mathcal{M}}(\ov \Xi_n)$ we have that
   \begin{equation}\label{eq:Step-c-easy-dir-Part3}
   \ov \Phi_n^\ast(\ov{\Q})
  =\sup_{\ov \gamma\in U_b(\ov\Omega)} (\E_{\ov{\Q}}[\ov \gamma]-\ov\Phi_n(\ov \gamma))=0,
   \end{equation}
   which is the first equality to show in \eqref{eq:conjugate}. 
   Therefore, it remains to show that $\Phi_n^\ast(\ov{\Q})=+\infty$ whenever $\ov{\Q}\notin\ov{\mathcal{M}}(\ov\Xi_n)$.
   To that end, let $\ov{\Q}$ be a finite Borel measure which is not in $\ov{\mathcal{M}}(\ov\Xi_n)$.
   Notice that $\ov \Phi_n(m)\leq m$ for every $m \in \R$ implies 
  that
  $\ov\Phi_n^\ast(\ov{\Q}) \geq \sup_m (m\,\ov{\Q}(\ov\Omega)-m)=\infty$
  whenever $\ov \Q(\ov\Omega)\neq 1$.
  Further, since $\ov \Xi_n$ is compact, there are continuous functions $\ov \gamma_k\colon\ov\Omega\to[0,k]$, $k \in \N$,
  which increase pointwise to $+\infty \ov 1_{\ov\Xi_n^c}$.
  Note that in particular $\ov \gamma_k=0$ on $\ov\Xi_n$ and hence $\ov\Phi_n(\ov \gamma_k)\leq 0$ for all $k$ by definition, and therefore
  \begin{align}
  \label{eq:conjugate.Z}
  \ov\Phi_n^\ast(\ov{\Q})
  \geq \sup_k (\E_{\ov{\Q}}[\ov \gamma_k]-\ov\Phi_n(\ov \gamma_k))
  \geq +\infty \ov{\mathbb{Q}} (\ov\Xi_n^c).
  \end{align}
  In particular, we see that $\ov\Phi_n^\ast(\ov{\Q})=\infty$ whenever $\ov{\mathbb{Q}} (\ov\Xi_n^c)\neq0$.
  Hence by the arguments already provided above, we may assume w.l.o.g.~that $\ov{\Q}\in \fP(\ov \Omega)$ and $\ov\Q(\ov\Xi_n)=1$. Since $\ov{\Q}\in \fP(\ov \Omega)$ but $\ov \Q \notin\ov{\mathcal{M}}(\ov\Xi_n)$, we have by definition that
  either $\ov S$ or $\ov{\mathbb{S}}$ is not
  a $\ov{\Q}$-$\ov{\mathbb{F}}$-local martingale. 
  In either case, by Lemma \ref{lem:polar} there is a function $\ov \gamma\in C_b(\ov\Omega)$ and 
  $\ov H\in\ov{\mathcal{H}}$ such that
  $\E_{\ov{\Q}}[\ov \gamma]>0$ and $\ov \gamma\leq (\ov H\cdot \ov S)_T$ or $\ov \gamma\leq (\ov H\cdot \ov{\mathbb{S}})_T$.
  This implies that 
  $\ov \Phi_n(m\ov \gamma)\leq 0 $ for all $m>0$. Therefore, as
  \begin{equation*}
  \ov\Phi_n^\ast(\ov{\Q})
  \geq
  \sup_{m>0}\big(\E_{\ov \Q}[m\ov \gamma]-\ov\Phi_n(m\ov \gamma)\big)
  \geq
  \sup_{m>0}\E_{\ov \Q}[m\ov \gamma],
  \end{equation*}
  we conclude that
  $\ov\Phi_n^\ast(\ov{\Q})=\infty$, if $\ov{\Q}\notin\ov{\mathcal{M}}(\ov\Xi_n)$. This together with \eqref{eq:Step-c-easy-dir-Part3} give
  \eqref{eq:conjugate} and in turn \eqref{eq:rep.phi.n.on.enlarged.space}.

  Step (c): Let  $\ov \xi\colon\ov\Omega\to[0,\infty]$ be such that there exists a sequence of bounded upper-semicontinuous functions
  $\ov\xi_n\colon\ov\Omega\to[0,\infty)$ such that $\ov\xi=\liminf_{n}\ov\xi_n$. For each $n$, define $\ov\xi_n':=\inf_{m\geq n}\ov\xi_m$.
  Then $\ov\xi_n'$ is bounded upper-semicontinuous and $\ov\xi=\sup_n \ov\xi_n'$. We show that 
  \begin{equation}\label{eq:rep2}
  \ov\Phi(\ov \xi)=\sup_n \ov \Phi_n(\ov\xi_n')=\sup_{\ov{\Q}\in\ov{\mathcal{M}}(\ov \Xi)} \E_{\ov{\Q}}[\ov \xi].
  \end{equation}
  Indeed, by \eqref{eq:weakduality}, \eqref{eq:weakduality-part2}, and \eqref{eq:rep.phi.n.on.enlarged.space} it holds
  \begin{align}
    \ov\Phi(\ov \xi)
    &\geq \sup_{\ov{\Q}\in\ov{\mathcal{M}}(\ov\Xi)} \E_{\ov{\Q}}[\ov \xi]
    \geq \sup_n\sup_{\ov{\Q}\in\ov{\mathcal{M}}(\ov \Xi_n)} \E_{\ov{\Q}}[\ov \xi_n'] 
     =\sup_n \ov\Phi_n(\ov \xi_n'). \label{eq:Step-d-ineq1}
  \end{align}
  On the other hand, let $\lambda>\sup_n\ov\Phi_n(\ov \xi_n')$ and $\varepsilon>0$ arbitrary.
  Then for every $n$ there exist $\ov H^n, \ov G^n$ in $\ov{\mathcal{H}}$ with
  \[\lambda+(\ov H^n\cdot \ov S)_T+(\ov G^n\cdot \ov{\mathbb{S}})_T\geq \ov \xi_n' \quad \mbox{on } \ov \Delta\cap\ov\Xi_n.\]
  Since $\ov  \xi_n'\geq 0$, by Lemma \ref{lem:integral.positive}, 
  there are $\ov H^{n,'},\ov G^{n,'}\in \ov{\mathcal{H}}$ such that both
  \begin{align*}
  \lambda+\varepsilon+(\ov H^{n,'}\cdot \ov S)_T+(\ov G^{n,'}\cdot \ov{\mathbb{S}})_T&\geq \ov \xi_n' \quad \text{on } \ov \Delta\cap\ov\Xi_n,\\
  \lambda+\varepsilon+(\ov H^{n,'}\cdot \ov S)_t+(\ov G^{n,'} \cdot \ov{\mathbb{S}})_t&\geq \ov 0 \quad \text{on } \ov \Delta\cap \ov \Xi
  \text{ for every }t.
  \end{align*}
  Taking the $\liminf$ over $n$, this implies that $\ov\Phi(\ov \xi)\leq \lambda+\varepsilon$, and therefore
  $\ov\Phi(\ov\xi)\leq \sup_n\ov\Phi_n(\ov \xi_n')$.
  In particular, all inequalities in \eqref{eq:Step-d-ineq1} are equalities and the proof is complete.
\end{proof}

Now we are able to provide the proof of Theorem~\ref{thm:main.stopped}.

\begin{proof}[Proof of Theorem \ref{thm:main.stopped}]
  Fix $\xi\colon C[0,T]\to[0,\infty]$ which can be written as $\xi=\ov\xi\circ\psi:=\liminf_n \ov\xi_n\circ\psi$ where
  $\ov\xi_n\colon \ov\Omega\to[0,\infty)$, $n \in \N$, are bounded and upper-semicontinuous.
  By  the primal transition result  in Lemma~\ref{lem:primal-transition} we know that
  \begin{equation*}
  \Phi(\xi)= \ov{\Phi}(\xi \circ \ov S).
  \end{equation*}
  Hence, using that $\psi \circ \ov S = \mathrm{\mathop{id}}_{\ov\Omega}$ on $\ov\Delta$ and that by definition $\ov\Phi (\ov\xi_1) = \ov \Phi(\ov\xi_2)$ for any $\ov{\xi}_1, \ov{\xi}_2:\ov \Omega \to (-\infty,\infty]$  such that $\ov{\xi}_1 = \ov{\xi}_2$ on $\ov \Delta$ yields
  \begin{equation*}
  \ov{\Phi}(\xi \circ \ov S)
  =
  \ov{\Phi}(\ov\xi \circ \psi \circ \ov S)
  =
  \ov{\Phi}(\ov\xi).
  \end{equation*}
  Now, the duality result on the enlarged space provided in Proposition~\ref{prop:lifted.dual.with.stopped.paths} ensures that
  \begin{equation*}
 \ov{\Phi}(\ov\xi)=\sup_{\ov \Q \in \ov{\cM}(\ov \Xi)}\E_{\ov \Q}[\ov \xi].
  \end{equation*}
  Since  $\ov \Q(\ov \Delta)=1$ for all $\ov \Q \in \ov{\cM}(\ov \Xi)$ we have
  \begin{equation*}
  \sup_{\ov \Q \in \ov{\cM}(\ov \Xi)}\E_{\ov \Q}[\ov \xi]
  =
  \sup_{\ov \Q \in \ov{\cM}(\ov \Xi)}\E_{\ov \Q}[\ov \xi\circ \psi \circ \ov S]
  =
  \sup_{\ov \Q \in \ov{\cM}(\ov \Xi)}\E_{\ov \Q}[ \xi \circ \ov S].
  \end{equation*}
  Then the dual transition result in Lemma~\ref{lem:dual-transition} ensures that
  \begin{equation*}
\sup_{\ov \Q \in \ov{\cM}(\ov \Xi)}\E_{\ov \Q}[ \xi \circ \ov S]
=
\sup_{\Q \in {\cM}(\Xi)}\E_{\Q}[ \xi],
  \end{equation*}
  which shows the desired result.
\end{proof}

\subsection{Proof of Theorem \ref{thm:main.stopped.liminf.clsoure}}

In this subsection we provide the proof of Theorem~\ref{thm:main.stopped.liminf.clsoure}. To that end, let $\Xi\subseteq \Omega$ be a prediction set  which satisfies Assumption~\ref{ass:A} and recall for every $\lambda>0$ the set
\[\mathcal{G}^\Xi_\lambda:=\liminf\text{-closure of }\{\lambda + (H\cdot S)_T: H\in\mathcal{H}
\text{ and } \lambda+(H\cdot S)_t\geq 0 \text{ on } \Xi \text{ for all }t \}.\]
Moreover, let $\ov \Xi_n$, $n \in \N$, be  the sets introduced in Lemma~\ref{lem:stopped.prediction.set} and denote for each $n$ the set $\Xi_n:=\{\omega \in \Omega\colon(\omega,\langle \omega \rangle)\in \ov \Xi_n\}$.

\begin{lemma}
\label{lem:liminf.contains.fat.S}
Let $\Xi \subseteq \Omega$ and let 
$H,G\in\mathcal{H}$ such that $\lambda +(H\cdot S)_t+ (G\cdot \mathbb{S})_t\geq0$ on $\Xi$
  for all $t$.
  Then for every $\varepsilon>0$ there exist $Y\in\mathcal{G}^\Xi_{\lambda+\varepsilon}$ such that
  $Y\geq \lambda +(H\cdot S)_T+ (G\cdot \mathbb{S})_T$
  on $\Xi$.
\end{lemma}
\begin{proof}
  For every $m$ define 
  $S^m:=\sum_{k=0}^\infty S_{\sigma^m_k} 1_{(\sigma^m_{k},\sigma^m_{k+1}]}$ with $(\sigma^m_k)$ begin the stopping
  times from Subsection \ref{subsec:setup}. By definition of $\Omega$, it holds 
  $\mathbb{S}_t(\omega)=\lim_m (S^m\cdot S)_t(\omega)$ uniformly in $t$ for every $\omega\in\Omega$.
  As $G$ is simple, also $(GS^m\cdot S)_t(\omega)\to (G\cdot \mathbb{S})_t(\omega)$ uniformly in $t$.
  For the stopping times    
  $\tau^m:=\inf\{t\geq 0 : \lambda+\varepsilon+((H+GS^m)\cdot S)_t\leq 0\}$ and $K^m:=(H+G S^m)1_{[0,\tau^m]}$
  it holds $Y^m_t:=\lambda+\varepsilon+(K^m\cdot S)_t\geq 0$ for all $t$. Define $Y_T:=\liminf_m Y_T^m$. Then, 
  by uniform convergence of $(GS^m\cdot S)$ to $(G\cdot \mathbb{S})$ 
  and as $\lambda +(H\cdot S)_t+ (G\cdot \mathbb{S})_t\geq0$ on $\Xi$ we have that $\lim_{m}\tau^m=\infty$ 
  on $\Xi$. This ensures that
  $\lambda +(H\cdot S)_T+ (G\cdot \mathbb{S})_T \leq Y_T$ 
  on $\Xi$. Therefore, it remains to show that $Y_T \in \mathcal{G}^\Xi_{\lambda+\varepsilon}$.
  
  To that end,
  for every $m$ and $l$ define $S^{m,l}:=\sum_{k=0}^l S_{\sigma^m_k} 1_{(\sigma^m_{k},\sigma^m_{k+1}]}$
  and $\tau^{m,l}:=\inf\{t\geq 0 :\lambda+\varepsilon+((H+G S^{m,l})\cdot S)_t\leq 0\}$.
  Since $K^{m,l}:=(H+G S^{m,l})1_{[0,\tau^{m,l}]}\in\mathcal{H}$ and $Y_t^{m,l}:=\lambda+\varepsilon+(K^{m,l}\cdot S)_t\geq0$
  for all $t$, it follows that $Y^{m,l}_T\in\mathcal{G}^\Xi_{\lambda+\varepsilon}$ for all $m,l$.
  Further, as we are working on the continuous paths, for every $\omega\in C[0,T]$ and $m$ there is $l_0=l_0(m,\omega)$ 
  such that $S^{m,l}(\omega)=S^{m}(\omega)$ for $l\geq l_0$; hence $Y_T^m=\lim_l Y_T^{m,l}$.
  This implies $Y_T=\liminf_m (\liminf_l Y_T^{m,l})$ which, by liminf-closedness of $\mathcal{G}^\Xi_{\lambda+\varepsilon}$,
  implies that $Y_T\in\mathcal{G}^\Xi_{\lambda+\varepsilon}$ and completes the proof.
\end{proof}

\begin{proof}[Proof of Theorem \ref{thm:main.stopped.liminf.clsoure}]
  Let $n \in \N$. By Lemma \ref{lem:integral.positive} and the identity~\eqref{eq:rep.phi.n.on.enlarged.space} in Step~(b) of the proof of Proposition \ref{prop:lifted.dual.with.stopped.paths}, we obtain for every upper-semicontinuous bounded $\ov\xi\colon\ov\Omega\to[0,\infty)$ that
  \begin{align*}
  \inf\left\{ \lambda\in\mathbb{R} : 
      \begin{array}{l}
      \text{there are } \ov H,\ov G\in\ov{\mathcal{H}} \text{ such that}\\
      \lambda +(\ov H\cdot \ov S)_t+ (\ov G\cdot \ov{\mathbb{S}})_t\geq 0 \text{ on } \ov \Delta\cap \ov \Xi,\\
      \lambda+(\ov H\cdot \ov S)_T + (\ov G\cdot \ov{\mathbb{S}})_T\geq \ov \xi \text{ on } \ov\Delta\cap\ov\Xi_n
      \end{array} \right\}
      =\ov\Phi_n(\ov \xi)
      =\sup_{ \ov{\mathbb{Q}} \in \ov{\mathcal{M}} (\ov\Xi_n) } \mathbb{E}_{ \ov{\mathbb{Q}} }[\ov\xi].
  \end{align*}
  This, together with the same arguments used for the proofs of the transition results  in Lemma~\ref{lem:primal-transition} and Lemma~\ref{lem:dual-transition} (see also Lemma~\ref{lem:transfer.integrals} and Lemma~\ref{lem:transfer.measures}) ensure that for every bounded and upper semicontinuous function $\xi\colon C[0,T]\to[0,\infty)$ we have   
  \begin{equation}
  \label{eq:pf-clo1}
  \begin{split}
  \Phi_n(\xi)&:=\inf\left\{ \lambda\in\mathbb{R}: 
      \begin{array}{l}
      \text{there are  } H,G\in\mathcal{H} \text{ such that}\\
      \lambda +(H\cdot S)_t+ (G\cdot \mathbb{S})_t\geq 0 \text{ on } \Xi,\\
      \lambda+(H\cdot S)_T + (G\cdot \mathbb{S})_T\geq  \xi \text{ on } \Xi_n
      \end{array} \right\}
      =\sup_{\mathbb{Q}\in\mathcal{M}(\Xi_n)} \mathbb{E}_{\mathbb{Q}}[\xi].
      \end{split}
  \end{equation}
  Moreover,
  by Lemma \ref{lem:liminf.contains.fat.S} one has for every $\xi\colon C[0,T]\to[0,\infty]$ (regardless of measurability) that
   \begin{equation}
   \label{eq:pf-clo2}
  \Phi_n^{\mathrm{cl}}(\xi):=\inf\big\{\lambda \in\mathbb{R} : \xi\leq Y \text{ on } \Xi_n \text{ for some } Y\in\mathcal{G}^\Xi_\lambda \big\}
  \leq \Phi_n(\xi).
  \end{equation}
  Furthermore, by Remark \ref{rem:weak.duality} one has for every Borel $\xi\colon C[0,T]\to[0,\infty]$  and $\mathbb{Q}\in\mathcal{M}(\Xi_n)$ that
  \begin{equation}\label{eq:pf-clo3}
  \mathbb{E}_{\mathbb{Q}}[\xi]\leq\Phi_n^{\mathrm{cl}}(\xi).
  \end{equation}
  Therefore, we conclude from \eqref{eq:pf-clo1}--\eqref{eq:pf-clo3} that for every $n \in \N$ we have
  \begin{equation}
 \label{eq:pf-clo-duality-n-usc}
  \Phi_n^{\mathrm{cl}}(\xi)
  =\sup_{\mathbb{Q}\in\mathcal{M}(\Xi_n)} \mathbb{E}_\mathbb{Q}[\xi]
  \quad\text{for every non-negative } \xi \in U_b(C[0,T]).
  \end{equation}
  
  Our next goal is to extend the duality result obtained in \eqref{eq:pf-clo-duality-n-usc} also for bounded non-negative $\xi$ which are Borel. The idea is to apply Choquet's capacitability theorem (in the functional form). To that end, fix $n \in \N$ and introduce for every $\xi\colon C[0,T] \to \R$ the following two functionals
  \begin{equation}\label{eq:choquet-0}
  \Upsilon_1 (\xi):= \Phi_n^{\mathrm{cl}}(\xi^+),  \quad \quad \quad 
  \Upsilon_2 (\xi):= \sup_{\mathbb{Q}\in\mathcal{M}(\Xi_n)} \mathbb{E}_\mathbb{Q}[\xi^+], 
  \end{equation} 
  where $\xi^+:=\max\{\xi,0\}$ denotes the positive part. Notice that for every non-negative $\xi\in U_b(C[0,T])$ we have by \eqref{eq:pf-clo-duality-n-usc} that
  \begin{equation}\label{eq:choquet-1}
 \Upsilon_1 (\xi)= \Upsilon_2 (\xi)
  \end{equation} 
  and our goal is to show that the above identity can be extended to all $\xi\colon C[0,T] \to [0,\infty)$ which  are bounded and Borel measurable. To see this observe first that
  as for every $\xi\in U_b(C[0,T])$ one has $\xi\circ\ov S\in U_b(\ov\Omega)$,
  it follows from step (b) in the proof of Proposition \ref{prop:lifted.dual.with.stopped.paths} 
  and \cite[Theorem 2.2]{BartlCheriditioKupper.17} that
  $\xi\mapsto \Upsilon_1 (\xi)$ and 
  $\xi\mapsto \Upsilon_2 (\xi)$
  are continuous from above on $U_b(C[0,T])$.
  Second, let $\xi_m\colon C[0,T]\to\mathbb{R}$, $m \in \N$,
  be a sequence of bounded Borel functions  which increase
  pointwise to $\xi$. Then it follows from interchanging two suprema that 
  $\sup_{\mathbb{Q}\in\mathcal{M}(\Xi_n)} \mathbb{E}_\mathbb{Q}[\xi^+]
  =\sup_m\sup_{\mathbb{Q}\in\mathcal{M}(\Xi_n)} \mathbb{E}_\mathbb{Q}[\xi^+_m]$. This implies that the functional $\Upsilon_2$ is continuous from below on the set of bounded Borel measurable functions $\xi\colon C[0,T]\to \R$.
  Moreover, using the liminf-closedness of $\mathcal{G}^\Xi_\lambda$, one directly verifies that also 
  $\Phi_n^{\mathrm{cl}}(\xi^+)=\sup_m\Phi_n^{\mathrm{cl}}(\xi_m^+)$. Indeed, while $\Phi_n^{\mathrm{cl}}(\xi^+)\geq\sup_m\Phi_n^{\mathrm{cl}}(\xi_m^+)$ follows by monotonicity of the functional, we see for the reverse inequality that for every $\lambda>\sup_m\Phi_n^{\mathrm{cl}}(\xi_m^+)$ that there exists a sequence  $Y^m \in \cG^\Xi_\lambda$, $m \in \N$, such that for each $m$ we have $ Y^m \geq \xi_m^+$ on $\Xi_n$. This implies that $\liminf_{m\to \infty} Y^m \geq \xi^+$ on $\Xi_n$.
  The liminf-closedness of $\mathcal{G}^\Xi_\lambda$ ensures that $Y:=\liminf_m Y^m \in \mathcal{G}^\Xi_\lambda$ and hence $\Phi_n^{\mathrm{cl}}(\xi^+)<\lambda$, which in turn guarantees the reverse inequality.
  This shows  that also $\Upsilon_2$ is continuous from below on the set of bounded Borel measurable functions $\xi\colon C[0,T]\to \R$.
  
  To sum up, we have shown that both $\Upsilon_1$ and $\Upsilon_2$ are continuous from above on $U_b(C[0,T])$ and continuous from below on the set of bounded Borel measurable functions $\xi\colon C[0,T]\to \R$. In other words, both $\Upsilon_1$ and $\Upsilon_2$
  are functional capacities in the sense of Choquet over the set of Borel functions, hence by 
   Choquet's capacitability theorem 
  (in the functional form as e.g.~in \cite[Proposition 2.1]{BartlCheriditioKupper.17}), we obtain for every $\xi:C[0,T]\to \R$ which is bounded and Borel that
  \begin{equation}
  \Upsilon_i(\xi)=\sup\big\{\Upsilon_i(X)\colon U_b(C[0,T])\ni X\leq \xi\big\}, \quad i=1,2.
  \end{equation}
  This and \eqref{eq:choquet-1} together with \eqref{eq:choquet-0}  imply that  for every  $\xi:C[0,T]\to [0,\infty)$ which is bounded and Borel measurable we have as desired that
  \begin{equation}
  \label{eq:Choquet-Application}
  \begin{split}
  \Phi_n^{\mathrm{cl}}(\xi)=\Upsilon_1(\xi)
  &=
  \sup\big\{\Upsilon_1(X)\colon U_b(C[0,T])\ni X\leq \xi\big\}\\
  &=
  \sup\big\{\Upsilon_2(X)\colon U_b(C[0,T])\ni X\leq \xi\big\}
  =
   \Upsilon_2(\xi)
   = \sup_{\mathbb{Q}\in\mathcal{M}(\Xi_n)} \mathbb{E}_\mathbb{Q}[\xi].
   \end{split}
  \end{equation}

  In a final step, let $\xi\colon C[0,T]\to[0,+\infty]$ be Borel. Then by Remark \ref{rem:weak.duality} we have that
  \begin{equation}
  \label{eq:Choquet-final-1}
  \sup_{\mathbb{Q}\in\mathcal{M}(\Xi)} \mathbb{E}_\mathbb{Q}[\xi]\leq \Phi^{\mathrm{cl}}(\xi).
  \end{equation}
  On the other hand,  the same arguments used to show that $\Upsilon_1$ is continuous from below on the set of Borel functions show that the liminf-closedness of $\mathcal{G}^\Xi_{\lambda}$  implies that 
  \begin{equation}
  	\label{eq:Choquet-final-2}
  \Phi^{\mathrm{cl}}(\xi)\leq\sup_n\Phi_n^{\mathrm{cl}}(\xi\wedge n).
  \end{equation}
 Therefore, we conclude from \eqref{eq:Choquet-Application} together with \eqref{eq:Choquet-final-1} and \eqref{eq:Choquet-final-2} that
  \[\sup_{\mathbb{Q}\in\mathcal{M}(\Xi)} \mathbb{E}_\mathbb{Q}[\xi]
  \leq \Phi^{\mathrm{cl}}(\xi)
  \leq \sup_n\Phi_n^{\mathrm{cl}}(\xi\wedge n)
  = \sup_n \sup_{\mathbb{Q}\in\mathcal{M}(\Xi_n)} \mathbb{E}_\mathbb{Q}[\xi\wedge n]
  \leq \sup_{\mathbb{Q}\in\mathcal{M}(\Xi)} \mathbb{E}_\mathbb{Q}[\xi]\]
  so that all inequalities are in fact equalities and the result now follows.
\end{proof}

\subsection{Proof of Theorem \ref{thm:main.with.Z}}
\label{subsec:proof.theorem.Z}
In this subsection we provide the proof of Theorem~\ref{thm:main.with.Z}, which follows the idea of the proof of \cite[Theorem~2.3]{BartlKupperProemelTangpi.17}. To that end, we fix a prediction set $\Xi \subseteq \Omega$ which satisfies Assumption~\ref{ass:B} and let $\ov Z\colon \ov \Omega \to [1,\infty]$ be the corresponding function. Define as in Remark~\ref{rem:comaprison-Ass-A-B} 
the set $\overline{\Xi}\subseteq C[0,T]\times C[0,T]$  by
\begin{equation*}
\overline{\Xi}:= \big\{(\omega,\nu) \in C[0,T]\times C[0,T]\colon \overline{Z}(\omega,\nu)<\infty\big\}.
\end{equation*}
Moreover, define $Z\colon\Omega \to [1,\infty]$ by $Z=\ov Z\circ \psi$ and assume that $\cM_Z(\Xi)$ is nonempty.
%
%
\begin{proof}[Proof of Theorem \ref{thm:main.with.Z}]
	The proof is divided in the following steps.
	
	Step (a): Fix $n\in\mathbb{N}$. For any $\ov \xi \colon \ov\Omega\to(-\infty,\infty]$  we define
	\begin{align*}
	\ov \Phi^{\ov Z}_n(\ov \xi)&:=\inf\left\{ \lambda\in\mathbb{R} : 
	\begin{array}{l}
	\text{there is $c>0$ and $\ov H,\ov G\in\ov{\mathcal{H}}$ such that}\\
	\text{$\lambda +(\ov H\cdot \ov S)_T + (\ov G\cdot \ov{\mathbb{S}})_T\geq -c \mbox{ on } \ov \Delta\cap\ov\Xi$}\\
	\text{$\lambda+(\ov H\cdot \ov S)_T + (\ov G\cdot \ov{\mathbb{S}})_T\geq \ov \xi -\ov Z/n\mbox{ on } \ov\Delta\cap\ov\Xi$}
	\end{array} \right\}.
	\end{align*}
	Then, for any $\ov{\Q}\in\ov{\mathcal{M}}_{\ov Z}(\ov \Xi)$ and any Borel function $\ov \xi:\ov \Omega \to (-\infty,\infty]$ which is bounded from below the following hold
	\begin{align}
	\ov \Phi^{\ov Z}_n(\ov \xi)&\geq \E_{\ov{\Q}}[\ov \xi]-\E_{\ov{\Q}}[\ov Z/n], \label{eq:weakduality-Z} \\
	\ov \Phi^{\ov Z}(\ov \xi)&\geq \E_{\ov{\Q}}[\ov \xi]. \label{eq:weakduality2-Z}
	\end{align}
	In particular, as $\ov{\mathcal{M}}_{\ov Z}(\ov \Xi)$ is nonempty by Lemma~\ref{lem:M-M-large-trans}, the functional $\ov \Phi^{\ov Z}_n$ is real-valued on $U_b$.
	
	Indeed,
	let $\lambda >\ov \Phi^{\ov Z}_n(\ov \xi)$ so that there exist
	$\ov H,\ov G\in\ov{\mathcal{H}}$ and $c>0$ such that the inequalities 
  $\lambda+(\ov H\cdot \ov S)_T+(\ov G\cdot \ov{\mathbb{S}})_T\geq -c$
	and $\lambda+(\ov H\cdot \ov S)_T + (\ov G\cdot \ov{\mathbb{S}})_T\geq \ov \xi -\ov Z/n$ hold on $\ov \Delta\cap\ov\Xi$.
	Notice that for each $\ov \Q \in \ov{\cM}_{\ov Z}(\ov \Xi)$, both $\ov S$ and $\ov{\mathbb{S}}$ are true $\ov \Q$-$\ov \F$-martingales. Indeed, for any $\ov \Q \in \ov{\cM}_{\ov Z}(\ov \Xi)$ Assumption~\ref{ass:B} ensures that  
	\begin{equation*}
	\E_{\ov \Q}\big[\textstyle{\sup_{0\leq t \leq T}} |\ov S_t|\big]
	 +
	 \E_{\ov \Q}\big[\textstyle{\sup_{0\leq t \leq T} |\ov V_t|}\big]
	 \leq  \E_{\ov\Q}[\ov Z]<\infty.
	\end{equation*}
	Since $\ov V = \langle \ov S \rangle^{\ov \Q} \ \ov \Q$-a.s.\ the Burkholder-Davis-Gundy inequality ensures that $\ov S$ is a square integrable $\ov \Q$-$\ov \F$-martingale and hence as $\ov{\mathbb{S}}=(\ov S^2-\ov S_0^2-\ov V)/2$ we also obtain that $\ov{\mathbb{S}}$ is a true $\ov \Q$-$\ov \F$-martingale.
	Therefore, as by Remark~\ref{rem:ass-A2-enlarged} both $\ov S$ and $\ov{\mathbb{S}}$ are $\ov \Q$-$\ov \F^{\ov\Delta}_+$-martingales and $\ov H$, $\ov G$ are simple integrands, we see that $\E_{\ov \Q}\big[(\ov H \cdot \ov S)_T\big] =\E_{\ov \Q}\big[(\ov G \cdot \ov{\mathbb{S}})_T\big]=0$. Hence
	$\lambda \geq \E_{\ov{\Q}}[\ov \xi]-\E_{\ov{\Q}}[\ov Z/n]$, which shows \eqref{eq:weakduality-Z}. To prove \eqref{eq:weakduality2-Z}, we can use the same arguments together with Fatou's lemma.
	
	Step (b): Each $\ov \Phi^{\ov Z}_n$ is continuous from above on $C_b(\ov\Omega)$, that is,
	for every sequence $(\ov \xi_k)$ in $C_b(\ov\Omega)$ which decreases pointwise to $0$, one
	has $\ov\Phi^{\ov Z}_n(\ov \xi_k)\downarrow \ov \Phi^{\ov Z}_n(0)$.
	
	Indeed, to that end, fix such a sequence $(\ov \xi_k)$ and an arbitrary $\varepsilon>0$.
	Then there exist $\ov H,\ov G\in\ov{\mathcal{H}}$ with $\lambda+(\ov H\cdot \ov S)_T+ (\ov G\cdot \ov{\mathbb{S}})_T\geq -c$ on $\ov \Delta\cap\ov\Xi$ for some $c\geq 0$ such that
	\[ 
	\varepsilon + \ov \Phi^{\ov Z}_n(0)+ (\ov H\cdot \ov S)_T+ (\ov G\cdot \ov{\mathbb{S}})_T +\ov Z/n\geq 0 \quad \mbox{ on } \ov \Delta\cap\ov\Xi. 
	\]
	Now define $b:=\sup_{\ov\omega\in\ov\Omega} \ov \xi_1(\ov\omega)-\varepsilon - \ov\Phi^{\ov Z}_n(0)+c$, so that
	\[b+\varepsilon +\ov\Phi^{\ov Z}_n(0)+ (\ov H\cdot \ov S)_T+(\ov G\cdot \ov{\mathbb{S}})_T\geq \ov \xi_1\quad \mbox{ on } \ov \Delta\cap\ov\Xi.\]
	As $\{\ov Z\leq b n\}\subseteq \ov\Omega$ is compact by assumption, 
	Dini's lemma yields $\ov \xi_k 1_{\{\ov Z\leq b n\}}\leq\varepsilon$ for all $k$ large enough. Hence for $k$ big enough we have on $\ov \Delta\cap\ov\Xi$
	\begin{align*}
	\ov \xi_k &\leq \ov \xi_k 1_{\{\ov Z\leq b n\}}+\ov \xi_1 1_{\{\ov Z>b n\}} \\
	&\leq \varepsilon+\big(\varepsilon+\ov\Phi^{\ov Z}_n(0)+(\ov H\cdot \ov S)_T+ (\ov G\cdot \ov{\mathbb{S}})_T+ \ov Z/n\big)1_{\{\ov Z> b n\}}\\
	&\leq 2\varepsilon+\ov \Phi^{\ov Z}_n(0)+ (\ov H\cdot \ov S)_T+ (\ov G\cdot \ov{\mathbb{S}})_T +\ov Z/n.
	\end{align*}
	Therefore, $\ov \Phi^{\ov Z}_n(\ov \xi_k)\leq \ov \Phi^{\ov Z}_n(0)+2\varepsilon$ for $k$ large enough which shows that 
	$\ov\Phi^{\ov Z}_n(\ov \xi_k)\downarrow \ov\Phi^{\ov Z}_n(0)$. 
	
	Step (c): We proceed to show that for each $n$ and every finite Borel measure $\ov\Q$ on $\ov\Omega$ one has
	\begin{align}
	\begin{split}
	\label{eq:conjugate-Z}
	\ov \Phi_n^{\ov Z,\ast}(\ov{\Q})
	:=\sup_{\ov \gamma\in C_b(\ov\Omega)} (\E_{\ov{\Q}}[\ov \gamma]-\ov \Phi^{\ov Z}_n(\ov \gamma))
	&=\sup_{\ov \gamma\in U_b(\ov\Omega)} (\E_{\ov{\Q}}[\ov \gamma]-\ov\Phi^{\ov Z}_n(\ov \gamma))\\
	&=\begin{cases}
	\E_{\ov{\Q}}[\ov Z]/n, &\text{if } \ov{\Q}\in\ov{\mathcal{M}}_{\ov Z}(\ov \Xi),\\
	\infty,&\text{else}.
	\end{cases}
	\end{split}
	\end{align}
	
	Indeed, 
	if $\ov{\Q}\in\ov{\mathcal{M}}_{\ov Z}(\ov \Xi)$, then by \eqref{eq:weakduality-Z} we have for every $\ov \gamma\in U_b(\ov\Omega)$  that
	$\E_{\ov{\Q}}[\ov \gamma]-\ov \Phi^{\ov Z}_n(\ov \gamma)\leq \E_{\ov{\Q}}[\ov Z/n]$. In particular, for every $\ov{\Q}\in\ov{\mathcal{M}}_{\ov Z}(\ov \Xi)$
	\begin{equation}\label{eq:Step-c-easy-dir-Z}
	\ov \Phi^{\ov Z,\ast}_n(\ov{\Q})\leq \sup_{\ov \gamma\in U_b(\ov\Omega)} (\E_{\ov{\Q}}[\ov \gamma]-\ov\Phi^{\ov Z}_n(\ov \gamma))\leq \E_{\ov{\Q}}[\ov Z/n].
	\end{equation}
	On the other hand, 
 since $\ov Z\geq 1$ is lower semicontinuous (as it has compact sublevel sets)
	there exists a sequence of non-negative functions 
	$\ov Z_k\in C_b(\ov\Omega)$ which increase pointwise to $\ov Z$.
	By definition $\ov\Phi^{\ov Z}_n(\ov Z_k/n)\leq 0$ for all $k$, hence for every $\ov{\Q}\in \fP(\ov \Omega)$
	\begin{align}
	\label{eq:conjugate.Z-Z}
	\ov\Phi^{\ov Z,\ast}_n(\ov{\Q})
	\geq \sup_k \big(\E_{\ov{\Q}}[\ov Z_k/n]-\ov\Phi^{\ov Z}_n(\ov Z_k/n)\big)
	\geq \E_{\ov{\Q}}[\ov Z/n].
	\end{align}
	Therefore, we conclude from \eqref{eq:Step-c-easy-dir-Z}
	and
	 \eqref{eq:conjugate.Z-Z}
	 that indeed for every $ \ov \Q \in \ov{\cM}_{\ov Z}(\ov \Xi)$ we have
	 \begin{equation}
	 \label{eq:conjugate-enlarge-Z-Part1}
	 \ov \Phi^{\ov Z,\ast}_n(\ov{\Q})= \sup_{\ov \gamma\in U_b(\ov\Omega)} (\E_{\ov{\Q}}[\ov \gamma]-\ov\Phi^{\ov Z}_n(\ov \gamma))= \E_{\ov{\Q}}[\ov Z/n].
	 \end{equation}
	 which is the first equality to be shown in \eqref{eq:conjugate-Z}. Therefore, it remains to show that for any Borel measure $\ov \Q$ on $\ov\Omega$ such that $ \ov \Q \notin \ov{\cM}_{\ov Z}(\ov \Xi)$, we have that $\ov \Phi^{\ov Z,\ast}_n(\ov{\Q})=\infty$.

	To that end, let $\ov \Q$ be a finite Borel measure on $\ov\Omega$ such that $ \ov \Q \notin \ov{\cM}_{\ov Z}(\ov \Xi)$. Note  that for every $m\in\mathbb{R}$ we have $\ov \Phi^{\ov Z}_n(m)\leq m$ by definition, so that
	$\ov\Phi^{\ov Z,\ast}_n(\ov{\Q}) \geq \sup_m (m\,\ov{\Q}(\ov\Omega)-m)=\infty$
	whenever $\ov{\Q}\notin  \fP(\ov \Omega)$.
	Moreover,  as $\{\ov Z =\infty\}=\ov \Xi^c$, it follows from \eqref{eq:conjugate.Z-Z} that $\ov\Phi^{\ov Z,\ast}_n(\ov{\Q})=\infty$ whenever $\ov \Q(\ov \Xi^c)>0$.
	Therefore, we may assume w.l.o.g.\ that $\ov{\Q}\in \fP(\ov \Omega)$, $\ov\Q(\ov\Xi)=1$, and due to \eqref{eq:conjugate.Z-Z} also that $\ov Z$ is integrable w.r.t.~$\ov{\Q}$.
	Therefore, by definition of $ \ov \Q \notin \ov{\cM}_{\ov Z}(\ov \Xi)$ we see that $\ov S$ or $\ov{\mathbb{S}}$ is not
	a $\ov{\Q}$-$\ov{\mathbb{F}}$-local martingale. 
	In either case, by Lemma \ref{lem:polar} there is a function $\ov \gamma\in C_b(\ov\Omega)$ and 
	$\ov H\in\ov{\mathcal{H}}$ such that
	$\E_{\ov{\Q}}[\ov \gamma]>0$ and $\ov \gamma\leq (\ov H\cdot \ov S)_T$ or $\ov \gamma\leq (\ov H\cdot \ov{\mathbb{S}})_T$.
	This implies that 
	$\ov \Phi^{\ov Z}_n(m\ov \gamma)\leq 0 $ for all $m>0$. Therefore, as
	\begin{equation*}
	\ov\Phi^{\ov Z,\ast}_n(\ov{\Q})
	\geq
	\sup_{m>0}\big(\E_{\ov \Q}[m\ov \gamma]-\ov\Phi^{\ov Z}_n(m\ov \gamma)\big)
	\geq
	\sup_{m>0}\E_{\ov \Q}[m\ov \gamma],
	\end{equation*}
	we conclude that
	$\ov\Phi^{\ov Z,\ast}_n(\ov{\Q})=\infty$, if $\ov{\Q}\notin\ov{\mathcal{M}}_{\ov Z}(\ov \Xi)$. This and  \eqref{eq:conjugate-enlarge-Z-Part1} hence indeed implies
	\eqref{eq:conjugate-Z}.
	As a consequence, we obtain from the non-linear Daniell-Stone theorem in \cite[Theorem~2.2]{BartlCheriditioKupper.17}	that
	\begin{equation}\label{eq:dualrep-Z}
	\ov\Phi^{\ov Z}_n(\ov \xi)=\sup_{\ov{\Q}\in\ov{\mathcal{M}}_{\ov Z}(\ov \Omega)}\big(  \E_{\ov{\Q}}[\ov \xi] -\E_\Q[\ov Z]/n\big)\quad\mbox{for all }\xi\in U_b(\ov\Omega).
	\end{equation}
	
	Step (d):
	Let  $\ov \xi\colon\ov\Omega\to(-\infty,\infty]$ be bounded from below such that there is a sequence of bounded upper-semicontinuous functions
	$\ov\xi_n\colon\ov\Omega\to\R$ with $\ov\xi=\liminf_{n}\ov\xi_n$. For each $n$, define $\ov\xi_n':=\inf_{m\geq n}\ov\xi_n$.
	Then $\ov\xi_n'$ is bounded upper-semicontinuous and $\ov\xi=\sup_n \ov\xi_n'$.
	We show that 
	\begin{equation}\label{eq:rep2-Z}
	\ov\Phi^{\ov Z}(\ov \xi)=\sup_n \ov \Phi^{\ov Z}_n(\ov\xi_n')=\sup_{\ov{\Q}\in\ov{\mathcal{M}}_{\ov Z}(\ov \Xi)} \E_{\ov{\Q}}[\ov \xi].
	\end{equation}

	Indeed, by \eqref{eq:weakduality2-Z} and \eqref{eq:dualrep-Z} it holds
	\begin{align}
	\ov\Phi^{\ov Z}(\ov \xi)
	&\geq \sup_{\ov{\Q}\in\ov{\mathcal{M}}_{\ov Z}(\ov \Xi)} \E_{\ov{\Q}}[\ov \xi]
	=\sup_{\ov{\Q}\in\ov{\mathcal{M}}_{\ov Z}(\ov \Xi)}\Big(\sup_n \big(\E_{\ov{\Q}}[\ov\xi_n'] -\E_{\ov{\Q}}[\ov Z]/n\big) \Big) \nonumber\\
	&=\sup_n\sup_{\ov{\Q}\in\ov{\mathcal{M}}_{\ov Z}(\ov \Xi)}\big(  \E_{\ov{\Q}}[\ov\xi_n'] -\E_\Q[\ov Z]/n\big)
	=\sup_n \ov\Phi^{\ov Z}_n(\ov\xi_n'). \label{eq:Step-d-ineq1-Z}
	\end{align}
	On the other hand, if $m:=\sup_n\ov\Phi^{\ov Z}_n(\ov\xi_n')$, 
	then for every $n$ there exists $\ov H^n$ and $\ov G^n$ in $\ov{\mathcal{H}}$ with
	\[m+\tfrac{1}{n}+(\ov H^n\cdot \ov S)_T+(\ov G^n\cdot \ov{\mathbb{S}})_T\geq \ov\xi_n'-\ov Z/n \quad \mbox{on } \ov \Delta\cap\ov\Xi.\] 
	Hence, $\lambda+(\ov H^n\cdot \ov S)_T+(\ov G^n\cdot \ov{\mathbb{S}})_T\geq -c\ov Z \mbox{ on } \ov \Delta\cap\ov\Xi$ for $c:=\|\ov \xi\wedge 0\|_\infty+m+2+|\lambda|$ and 
	\[m+ \liminf_{n\to \infty}\big( (\ov H^n\cdot \ov S)_T+ (\ov G^n\cdot \ov{\mathbb{S}})_T\big) \geq \liminf_{n \to \infty} \big(\ov\xi_n'-\ov Z/n-\tfrac{1}{n}\big)=\ov \xi \quad\text{on }
	\ov \Delta\cap\ov\Xi,\]
	where we use that $\ov \Delta\cap\ov\Xi \subseteq\{\ov Z<\infty\}$. 
	Therefore $\ov \Phi^{\ov Z}(\ov \xi)\leq m$, hence together with \eqref{eq:Step-d-ineq1-Z} 	we obtain that  
	indeed \eqref{eq:rep2-Z} holds true,
	and the infimum in the definition of $\ov \Phi^{\ov Z}(\ov \xi)$ is attained whenever $\ov \Phi^{\ov Z}(\ov \xi)<\infty$.

	Step (e): Let  $\xi\colon C[0,T]\to(-\infty,\infty]$ be bounded from below which can be written as $\xi=\ov\xi\circ\psi:=\liminf_n \ov\xi_n\circ\psi$ where
	$\ov\xi_n\colon \ov\Omega\to\R$, $n \in \N$, are bounded and upper-semicontinuous.  Notice that the fact that $Z=\ov Z \circ \psi$ on $\Omega$ and $Z\circ \ov S = \ov Z \circ \psi \circ \ov S =\ov Z $ on $\ov \Delta$ together with the same arguments used for the proofs of the transition results in Lemma~\ref{lem:primal-transition} and Lemma~\ref{lem:dual-transition} (see also Lemma~\ref{lem:transfer.integrals} and Lemma~\ref{lem:transfer.measures}) ensure that the same primal transition result obtained in Lemma~\ref{lem:primal-transition} also holds true with respect to $\Phi^{Z}$ and $\ov \Phi^{\ov Z}$ and that the same dual transition result obtained in Lemma~\ref{lem:dual-transition} also hold true  with respect to $\cM_{Z}(\Xi)$ and $\ov{\cM}_{\ov Z}(\ov \Xi)$. Therefore, the duality result on the enlarged space obtained in \eqref{eq:rep2-Z} together with the same arguments as in the proof of Theorem~\ref{thm:main.stopped} imply that
	\begin{equation*}
	\begin{split}
	\Phi^{Z}(\xi)
	= 
	\ov{\Phi}^{\ov Z}(\xi \circ \ov S)
	=
	\ov{\Phi}^{\ov Z}(\ov\xi)
	=
	\sup_{\ov \Q \in \ov{\cM}_{\ov Z}(\ov \Xi)}\E_{\ov \Q}[\ov \xi]
	=
	\sup_{\ov \Q \in \ov{\cM}_{\ov Z}(\ov \Xi)}\E_{\ov \Q}[ \xi \circ \ov S]
	=
	\sup_{\Q \in {\cM}_{Z}(\Xi)}\E_{\Q}[ \xi].
	\end{split}
	\end{equation*}
		Moreover, we know from Step (d) that the infimum in the definition of $\ov \Phi^{\ov Z}(\xi \circ \ov S)$ is attained whenever $\ov \Phi^{\ov Z}(\xi \circ \ov S)<\infty$. 
	Therefore, we see from Lemma~\ref{lem:transfer.integrals} together with the fact that $\psi(\Xi)= \ov \Delta \cap \ov \Xi$ and $\ov{S}(\ov \Delta \cap \ov \Xi)=\Xi$ 
	that also the infimum in the definition of $\Phi^{Z}(\xi)$ is attained whenever $\Phi^Z(\xi)<\infty$. The proof is thus complete.
\end{proof}

We finish this section with the proof of Proposition \ref{prop:sufficient.assumption.for.prediction.set}.

\begin{proof}[Proof of Proposition \ref{prop:sufficient.assumption.for.prediction.set}]
  Due to \eqref{eq:cond:integ:le:M-MZ}, a version of Kolmogorov's continuity criterion (see e.g.~\cite[Theorem A.1]{BartlKupperProemelTangpi.17}) ensures that for any $\alpha\in(0,1/4)$ there is a constant $C$ such that   $\sup_{\mathbb{Q}\in\mathcal{M}(\Xi)} \mathbb{E}_{\mathbb{Q}}[\|S\|_\alpha^4+\|\langle S\rangle\|_\alpha^4]\leq C$, where we denote $\|\omega\|_\alpha:=\sup_{s\neq t}|\omega(t)-\omega(s)|/|t-s|^\alpha$.
  Using the elementary inequality $(a+b)^4\leq 8 a^4+8 b^4$ implies that $\sup_{\mathbb{Q}\in\mathcal{M}(\Xi)} \mathbb{E}_{\mathbb{Q}}[\|S\|_{\alpha}^4 +\|\langle S\rangle \|_{\alpha}^4]\leq C'$ for all $n\geq 1/\alpha$ and some new constant $C'$.
  Therefore, by Markov's inequality, there is a sequence $(a_n)_{n \in \N}$ which increases to $\infty$ such that $\sup_{\mathbb{Q}\in\mathcal{M}(\Xi)} \mathbb{Q}(\|S\|_{\alpha} + \|\langle S\rangle \|_{\alpha}> a_n)\leq 1/n^3$. 
  Then the function 
  \[C[0,T]\times C[0,T]\ni(\omega,\nu) \mapsto\overline{Y}(\omega,\nu):=\inf\!\big\{ n\geq 1/\alpha : \|\omega\|_{\alpha}+\|\nu\|_{\alpha}\leq a_{n+1}\big\} \in \N
  \cup\{+\infty\}\]
  has compact level sets by the Arzel\`a-Ascoli theorem (see also \cite[Lemma~3.1]{BartlKupperProemelTangpi.17}) and satisfies $\{\overline{Y}<\infty\}\subset C^{\textup{H\"older}}[0,T]\times C^{\textup{H\"older}}[0,T]$.
  Let $Y(\omega):=\ov Y(\omega,\langle\omega\rangle)$ for $\omega\in\Omega$.
  Then, as $\{Y=n\}\subseteq\{\|S\|_{\alpha} + \|\langle S\rangle \|_{\alpha}> a_n\}$ and the latter set has probability less than $n^{-3}$ under all $\mathbb{Q}\in\mathcal{M}(\Xi)$, one gets
  \begin{equation}\label{eq:le:M-MZ-char}
  \sup_{\mathbb{Q}\in\mathcal{M}(\Xi)} \mathbb{E}_{\mathbb{Q}}[Y]<\infty.
  \end{equation}
  Further, as $\|\omega\|_\infty\leq |\omega(0)| + T^\alpha \|\omega\|_\alpha$, the assumptions yield $\sup_{\mathbb{Q}\in\mathcal{M}(\Xi)} \mathbb{E}_{\mathbb{Q}}[\|S\|_\infty+\|\langle S\rangle\|_\infty]<\infty$.
  Therefore, possibly replacing $\overline{Y}$ by $\overline{Y}(\omega,\nu)+\|\omega\|_\infty+\|\nu\|_\infty$, one may assume that $\overline{Y}(\omega,\nu)\geq \|\omega\|_\infty+\|\nu\|_\infty$.

  In a final step define $\overline{Z}:=\overline{Y}+\infty 1_{\overline{\Xi}^c}$ so that the assumption that $\ov \Xi \subseteq C^{\textup{H\"older}}[0,T]\times C^{\textup{H\"older}}[0,T]$ guarantees that $\overline{\Xi}=\{\overline{Z}<\infty\}$. 
  Therefore, we can conclude that Assumption~\ref{ass:B} is satisfied with respect to the constructed function $\ov Z$. 
  Moreover, \eqref{eq:le:M-MZ-char} assures that $\mathcal{M}(\Xi)=\mathcal{M}_Z(\Xi)$ for $Z\colon\Omega \to [1,\infty]$ defined by $Z(\omega):=\ov Z(\omega,\langle \omega \rangle)$, $\omega \in \Omega$.
\end{proof}

\section{Technical Results}
\label{sec:technical-Results}

In this section, if not explicitly stated otherwise, we use the setting of Subsection~\ref{subsec:setup}.

The following results are similar to ones in \cite{BartlKupperProemelTangpi.17}, whose proofs we provide for the sake of completeness.

\begin{lemma}
\label{lem:stopped.prediction.set}
Let $\Xi \subseteq \Omega$ be a prediction set which satisfies Assumption~\ref{ass:A}. Then there is an increasing sequence of nonempty compact sets 
  $\ov{\Xi}_n\subseteq C[0,T]\times C[0,T]$, $n \in \N$, $($i.e.\ $\forall n\colon\ov{\Xi}_n\subseteq\ov{\Xi}_{n+1})$  such that $\ov\Xi=\bigcup_n\ov\Xi_n$ 
  and for all $n$ we have that $\ov\omega \in \ov \Xi_n$ implies that $\ov{\omega}^t\in\ov{\Xi}_n$ for every $t$.
\end{lemma}
\begin{proof}
  Due to Assumption~\ref{ass:A} we know that $\ov\Xi=\bigcup_n\ov\Xi_n'$ for compacts $\ov\Xi_n'$ which  w.l.o.g.\  can be chosen to be increasing and nonempty.
  As $\ov\Xi_n:=\{ \ov\omega^t : \omega\in\ov\Xi_n' \text{ and } t\in[0,T]  \}\subseteq\ov\Xi$
  is also compact, the claim follows. 
\end{proof}

\begin{lemma}
\label{lem:integral.positive}
Let $\Xi \subseteq \Omega$ be a prediction set which satisfies Assumption~\ref{ass:A} and let $(\ov \Xi_n)_{n \in \N}$ be the sets introduced in Lemma~\ref{lem:stopped.prediction.set}.
 Moreover, let $n \in \N$, 
   $\ov H,\ov G\in\ov{\mathcal{H}}$, and $\lambda \geq 0$ such that
  $\lambda+(\ov H\cdot \ov S)_T+ (\ov G \cdot \ov{\mathbb{S}})_T\geq 0$ on $\ov\Delta\cap\ov{\Xi}_n$.
  Then for every $\varepsilon>0$ there are $\ov H^{'},\ov G'\in\ov{\mathcal{H}}$ such that
  \begin{align*}
  (\ov H'\cdot \ov S)_T+ (\ov G' \cdot \ov{\mathbb{S}})_T
  &=(\ov H\cdot \ov S)_T+ (\ov G \cdot \ov{\mathbb{S}})_T
  \quad\text{on } \ov\Delta\cap \ov \Xi_n,\\
  \lambda+\varepsilon + (\ov H'\cdot \ov S)_t+(\ov G'\cdot\ov{\mathbb{S}})_t
  &\geq 0 
  \quad\text{on }\ov\Delta\cap \ov \Xi\text{ for all } t.
  \end{align*}
\end{lemma}
\begin{proof}
  Fix $\ov \omega\in\ov\Delta\cap\ov\Xi_n$ and let $t$ be arbitrary. Since $\ov \omega^t\in\ov\Xi_n\cap\ov\Delta$, one obtains 
  \[ \lambda+(\ov H\cdot \ov S)_t(\ov\omega)+ (\ov G \cdot \ov{\mathbb{S}})_t(\ov\omega)
  =\lambda+(\ov H\cdot \ov S)_T(\ov\omega^t)+ (\ov G \cdot \ov{\mathbb{S}})_T(\ov\omega^t)
  \geq 0,\]
  see \cite[Lemma 4.6]{BartlKupperProemelTangpi.17} for more details.
  Now define the stopping time 
  $\ov\sigma:=\inf\{t\geq 0 : \lambda + \varepsilon+(\ov H\cdot \ov S)_t+ (\ov G \cdot \ov{\mathbb{S}})_t \leq 0\}$ 
  and then set $\ov H':=\ov H 1_{[0,\ov\sigma]}\in\ov{\mathcal{H}}$ and similarly
  $\ov G':=\ov G 1_{[0,\ov\sigma]}\in\ov{\mathcal{H}}$ to obtain the result.
\end{proof}

\begin{lemma}\label{lem:approx.Ft}
  Let $[a,b]\subseteq \R$. For every $\ov{\Q}\in\mathfrak{P}(\ov\Omega)$ and every
  $\ov{\mathcal{F}}_t$-measurable function $\ov h\colon\ov\Omega\to[a,b]$ 
  there are continuous $\ov{\mathcal{F}}_t$-measurable functions 
  $\ov h_k\colon \ov\Omega\to[a,b]$ which converge $\ov{\Q}$-almost surely to $\ov h$.
\end{lemma}
\begin{proof}
  First notice that both $\ov\Omega$ and
  $\ov\Omega^t:=\{\ov\omega|_{[0,t]}: \ov \omega\in\ov\Omega\}$, each
  endowed with the sup-norm, are Polish spaces.
  Now as
  $\ov{\mathcal{F}}_t=\{\ov \pi^{-1}(\ov B) : \ov B\subseteq\ov \Omega^t\text{ is Borel}\}$
  where $\ov \pi\colon\ov\Omega\to \ov\Omega^t$
  is defined by $\ov \pi(\ov\omega)=\ov\omega|_{[0,t]}$, 
one has $\ov h=\ov h^t\circ\pi$ for some Borel function $\ov h^t\colon\ov\Omega^t\to\mathbb{R}$.
  Since $\ov\Omega^t$ is Polish and $\ov{\Q}^t:=\ov{\Q}\circ\ov \pi^{-1}$ a probability measure thereon,
  there are continuous functions $h_k^t\colon\ov\Omega^t\to\mathbb{R}$ such that
  $\ov h_k^t\to \ov h^t$ $\ov{\Q}^t$-almost surely.
  Therefore, we can define $\ov h_k\colon\ov\Omega\to\mathbb{R}$ by $\ov h_k:=\ov h_k^t\circ\ov \pi$.
\end{proof}

\begin{lemma}
\label{lem:S.stopped.lsc}
Let here $\Omega$ be any metric space and $X:\Omega \to C[0,T]$ continuous.
 Fix $0\leq s< t\leq T$, $m> 0$, and define 
  \[ \tau:=\inf\{r\geq s: X_r>m \text{ or } X_r\leq -m\}\wedge T.\]
  Then the function 
  $\omega\mapsto X_{\tau(\omega)\wedge t}(\omega)$ is lower semicontinuous.
\end{lemma}

\begin{proof}
  Define $\tau_+:=\inf\{r\geq s: X_r>m \}\wedge T$ and $\tau_-:=\inf\{r\geq s: X_r\leq -m\}\wedge T$,
  and note that $\tau=\tau_+\wedge \tau_-$. Moreover, fix ${\omega}$ and a sequence 
  $({\omega}_n)$ such that $\omega_n\to \omega$.
  
  First, we claim that
  \begin{equation}\label{eq:limsup-liminf-stopping-time}
\limsup_n \tau_+({\omega}_n)\leq\tau_+({\omega})
    \quad\text{and}\quad 
    \liminf_n \tau_-({\omega}_n)\geq\tau_-({\omega}).
  \end{equation}
  
  Indeed, for the first inequality, assume without loss of generality that $r:=\tau_+({\omega})<T$.  Fix any $\varepsilon>0$.
  Then, by definition,
  there is $\delta\in(0,\varepsilon)$ such that
  $X_{r+\delta}({\omega})>m$. Since ${\omega}\mapsto X({\omega})$ is continuous, 
  $X_{r+\delta}({\omega}_n)>m$ for eventually all $n$, 
  showing that $\tau_+({\omega}_n)\leq r+\varepsilon$ for eventually all $n$. 
  As $\varepsilon>0$ was arbitrary, the first inequality of the claim follows. 
  
  To  see the second inequality of the claim, we may assume without loss of generality that $r:=\tau_-({\omega})>s$.
  Then necessarily $X_u({\omega})>-m$ for $u\in[s,r)$. 
  Fix $\varepsilon>0$. By continuity of $t\mapsto X_t({\omega})$, there exists $\delta>0$
  such that  $X_u({\omega})\geq-m+\delta$ for $u\in[s,r-\varepsilon]$.
  Further, due to continuity of ${\omega} \mapsto X({\omega})$, 
  it follows that $X_u({\omega}_n)\geq -m+\delta/2$  for $u\in[s,r-\varepsilon]$
  for eventually all $n$. Therefore $\tau_-({\omega}_n)\geq r-\varepsilon$ for eventually all $n$
  and as $\varepsilon$ was arbitrary, the second part of the claim follows.
  
  Next, to prove the lower semicontinuity of $X^\tau_t$, we distinguish between several cases:
  
  (a) If $X^\tau_t({\omega})>m$, then by the continuity of the paths of $X.(\omega)$, $\tau({\omega})=\tau_+({\omega})=s$ and 
  $X_s({\omega})>m$. By continuity of $X_s(\cdot)$ on $\Omega$, $X_s({\omega}_n)>m$ and $\tau_+({\omega}_n)=s$ 
  for eventually all $n$, hence 
  $\lim_n X^\tau_t({\omega}_n)
  =\lim_n X_s({\omega}_n)
  =X_s({\omega})
  =X^\tau_t({\omega})$.
  
  (b) If $X^\tau_t({\omega})=m$, then either $\tau_+({\omega})<t$ or $\tau_+({\omega})\geq t$.
  In the first case it follows that $\tau_+({\omega})<\tau_-({\omega})$ so that by \eqref{eq:limsup-liminf-stopping-time}
  $\tau_+({\omega}_n)<\tau_-({\omega}_n)$ and $\tau_+({\omega}_n)<t$ 
  for eventually all $n$ 
  and therefore
  \[ \liminf_{n\to \infty} X^\tau_t({\omega}_n)
    =\liminf_{n\to \infty} X_{\tau_+({\omega}_n)}({\omega})
    =  m
    =X^\tau_t({\omega}). \]
  On the other hand, if $\tau_+({\omega})\geq t$, then $X_t({\omega})=m$
  and $X_r({\omega})>-m$ for $r\in[s,t]$.
  This implies that $\tau_-(\omega_n)\geq t$ for eventually all $n$ and therefore
  \[  \liminf_n X^\tau_t({\omega}_n)
    =\liminf_n X_{t\wedge \tau_+({\omega}_n)}({\omega}_n)
    = m 
    =X^\tau_t({\omega}). \]
  
  (c) If $X^\tau_t({\omega})\in (-m,m)$, then either $\tau({\omega})>t$ or $\tau({\omega})=T$ 
  (in which case necessarily $t=T$). 
  In the latter case it follows that $X_r({\omega})>-m$ for $r\in[s,T]$, hence $\tau_-({\omega}_n)=T$ 
  for eventually all $n$ and thus
  \[ \liminf_n X^\tau_t({\omega}_n)
    =\liminf_n X_{t\wedge \tau_+({\omega}_n)}({\omega}_n)
    \geq X_t({\omega})
    =X^\tau_t({\omega}).\]
  If $\tau({\omega})>t$, then again $\tau_-({\omega}_n)>t$ for eventually all $n$ so that the same argument shows
  that $\liminf_n X^\tau_t({\omega}_n)\geq X^\tau_t({\omega})$.
  
  (d) If $X^\tau_t({\omega})=-m$, then $X_s({\omega})\geq -m$. 
  Assume that $\liminf_n X^\tau_t({\omega}_n)<-m$.
  Then there is a subsequence still denoted by $({\omega}_n)$ such that 
  $\tau({\omega}_n)=\tau_-({\omega}_n)=s$ for eventually all $n$. 
  However, this contradicts 
  $\liminf_n X^\tau_t({\omega}_n)
  =\lim_n X_s({\omega}_n)
  =X_s({\omega})\geq -m$.
  
  (e) If $X^\tau_t({\omega})<-m$, then $\tau_-({\omega})=s$ and $X_s({\omega})<-m$. 
  This implies $X_s({\omega}_n)<-m$ and therefore
  $\tau_-({\omega}_n)=s$ for eventually all $n$, so that 
  $\lim_n X^\tau_t({\omega}_n)
  =\lim_n X_s({\omega}_n)
  =X_s({\omega})
  =X^\tau_t({\omega})$.
\end{proof}

\begin{proposition}
\label{lem:polar}
  Let $\ov X$ be either $\ov S$ or $\ov{\mathbb{S}}$ and fix a $\ov{\Q}\in\mathfrak{P}(\ov\Omega)$.
  If $\ov X$ is not a $\ov{\Q}$-$\ov{\F}$-local martingale, then there exists
  $\ov \gamma\in C_b(\ov\Omega)$ and $\ov H\in\ov{\mathcal{H}}$ such that 
  $\ov \gamma\leq (\ov H\cdot \ov X)_T$ and $\E_{\ov{\Q}}[\ov \gamma]>0$.
\end{proposition}

\begin{proof}
Consider the set
  \begin{equation*}
 \ov{\Gamma}:=\{\ov \gamma\in C_b(\ov\Omega) : \ov \gamma\leq (\ov H\cdot \ov X)_T \text{ for some } \ov H\in\ov{\mathcal{H}}\}.
  \end{equation*}
  We prove that if $\E_{\ov{\Q}}[\ov \gamma]\leq 0$ for all 
  $\ov \gamma\in \ov\Gamma$, 
  then $\ov X$ is a $\ov{\Q}$-$\ov{\F}$-local martingale with localizing sequence 
  \begin{equation*}
  \ov\tau_m:=\inf\{ t\geq 0 : |\ov X_t|\geq m\}\wedge T,
  \end{equation*}
  i.e.~for every $m\in\mathbb{N}$, the stopped process
  \[
    \ov X^{\ov \tau_m}_t:=\ov X_{t\wedge \ov\tau_m}
  \]
  is a $\ov{\Q}$-$\ov{\F}$-martingale. Fix $m\in\mathbb{N}$ and write $\ov\tau:=\ov\tau_m$. First let us show that $\ov X^{\ov \tau}$ is an $\ov \F$-submartingale. To that end, let $0\leq s< t\leq T$, and define 
  \begin{align*}
    \ov\sigma&:=\inf\{ r\geq  s : |\ov X_r|\geq m \}\wedge T,\\
    \ov\sigma_\varepsilon&:=\inf\{ r\geq s :  \ov X_r> m-\varepsilon \text{ or } \ov X_r \leq -m+\varepsilon\}\wedge T
  \end{align*}
  for $0<\varepsilon \leq 1$. 
  Since both $\ov \tau$ and $\ov \sigma$ are  hitting times of a closed set and $\ov X$ is continuous, 
    they are $\ov \F$-stopping times, whereas $\ov\sigma_\varepsilon$ is a $\ov\F_+$-stopping time.
    
    Now, fix an arbitrary $\ov \cF_s$-measurable function $\ov h: \ov \Omega \to [0,1]$. 
    Notice that $\ov \sigma=\ov \tau$ on $\{\ov \tau\geq s\}$, so that 
      $1_{\{\ov\tau\geq s\}}(\ov X_t^{\ov \sigma}-\ov X_s)=\ov X^{\ov \tau}_t-\ov X_s^{\ov \tau}$. Moreover,  $\ov \sigma_\varepsilon$ increases to $\ov \sigma$ as $\varepsilon$ tends to $0$, and therefore $\ov X_t^{\ov \sigma_\varepsilon}\to \ov X_t^{\ov \sigma}$ by continuity of $\ov X$.  Since additionally $|\ov X_t^{\ov \sigma_\varepsilon}-\ov X_s|\leq 2m$, this shows that
      \begin{equation*}
       \E_{\ov{\Q}}[\ov h(\ov X^{\ov \tau}_t-\ov X_s^{\ov \tau})]
      =
      \E_{\ov \Q}[ \ov h\,1_{\{\ov \tau\geq s\}}\,(\ov X_t^{\ov \sigma}-\ov X_s)] 
              =\lim_{\varepsilon\to0}  \E_{\ov{\Q}}[ \ov h\,1_{\{\ov \tau\geq s\}}\,(\ov X_t^{\ov \sigma_\varepsilon}-\ov X_s)].
      \end{equation*}
 Recall that $\ov g:=\ov h1_{\{\ov \tau\geq s\}}\colon\ov\Omega\to[0,1]$ is 
   $\ov{\mathcal{F}}_s$-measurable. By Lemma \ref{lem:approx.Ft}, there exists a sequence of continuous $\ov{\mathcal{F}}_s$-measurable 
     functions $\ov g_k\colon\ov\Omega\to[0,1]$
     which converge $\ov{\Q}$-almost surely to $\ov g$.  By Lemma~\ref{lem:S.stopped.lsc} the function 
  $\ov\omega\mapsto \ov X_{t\wedge \ov\sigma_\varepsilon(\ov\omega)}(\ov\omega)$ is lower semicontinuous 
  for every $\varepsilon$.
  In particular,  for every fixed $k$
  it holds that
  \[   \ov \omega \mapsto (\ov H\cdot \ov X)_T(\ov \omega)\text{ is lower semicontinuous, where }
  \ov H:=\ov g_k 1_{(s,\ov \sigma_\varepsilon\wedge t]}\in\ov{\mathcal{H}}. \] 
  Since additionally $|\ov X_t^{\ov \sigma_\varepsilon}-\ov X_s|\leq 2m$ and $\ov\Omega$ is a Polish space, 
  there exists a sequence of continuous functions $\ov \gamma_n\colon\ov\Omega\to[-2m,2m]$ such that 
  $\ov \gamma_n\leq (\ov H\cdot \ov X)_T$ and $\ov \gamma_n$ increases pointwise to $(\ov H\cdot \ov X)_T$. 
  Therefore $\ov \gamma_n\in \ov \Gamma$, hence by assumption 
  \[   \E_{\ov{\Q}}[ \ov g_k(\ov X_t^{\ov\sigma_\varepsilon}-\ov X_s)] 
  = \E_{\ov{\Q}}[ (\ov H\cdot \ov X)_T] 
  =\sup_n \E_{\ov{\Q}}[\ov \gamma_n] 
  \leq 0.\]
  We conclude that
  \begin{equation*}
  \E_{\ov{\Q}}[\ov h(\ov X^{\ov \tau}_t-\ov X_s^{\ov \tau})]
  =
  \lim_{\varepsilon\to0}  \E_{\ov{\Q}}[ \ov h\,1_{\{\ov \tau\geq s\}}\,(\ov X_t^{\ov \sigma_\varepsilon}-\ov X_s)]
  =
  \lim_{\varepsilon\to0}  \lim\limits_{k \to \infty}\E_{\ov{\Q}}[ \ov g_k\,(\ov X_t^{\ov \sigma_\varepsilon}-\ov X_s)]\leq 0,
  \end{equation*}
  which implies $\ov{\Q}$-almost surely 
  $\E_{\ov{\Q}}[\ov X^{\ov \tau}_t|\ov{\mathcal{F}}_s]\leq \ov X^{\ov \tau}_s$, 
  hence $\ov X^{\ov \tau}$ is a $\ov{\Q}$-$\ov{\F}$-supermartingale. 

  By similar arguments one can also show that $\ov X^{\ov \tau}$ is a 
  $\ov{\Q}$-$\ov{\F}$-submartingale and thus a $\ov{\Q}$-$\ov{\F}$-martingale. 
\end{proof}

\vspace{1em}
{\bf Acknowledgment:}
Part of this research was carried out while the first author was visiting the Shanghai Advanced Institute of Finance and the School of Mathematical Sciences at the Shanghai Jiao Tong University in China, and he would like to thank Samuel Drapeau for his hospitality. 
This author also acknowledges financial support by the Vienna Science and Technology Fund (WWTF) through project VRG17-005 and by the Austrian Science Fund (FWF) under grant Y00782.
{The third author gratefully acknowledges the financial support by the NAP Grant and  by the Swiss National Foundation Grant  SNF~200020$\_$172815.}



\newcommand{\dummy}[1]{}


\end{document}